\documentclass[nofootinbib,superscriptaddress,twocolumn]{revtex4-2}
\usepackage[utf8]{inputenc}
\usepackage[T1]{fontenc}
\usepackage{times}
\usepackage{graphicx}
\usepackage{amsmath,amsfonts,amssymb,amscd,mathtools,amsthm}
\usepackage{hyperref}
\usepackage[capitalise]{cleveref}
\usepackage[dvipsnames]{xcolor}
\usepackage{framed}
\usepackage{physics}
\usepackage[normalem]{ulem}
\hypersetup{
	bookmarksnumbered=true, 
	unicode=false, 
	pdfstartview={FitH}, 
	pdftitle={}, 
	pdfauthor={}, 
	pdfsubject={}, 
	pdfcreator={}, 
	pdfproducer={}, 
	pdfkeywords={}, 
	pdfnewwindow=true, 
	colorlinks=true, 
	linkcolor=NavyBlue, 
	citecolor=NavyBlue, 
	filecolor=NavyBlue, 
	urlcolor=NavyBlue 
}
\usepackage[a4paper, margin=0.7in]{geometry}

\usepackage{enumitem}
\usepackage{ifthen}

\theoremstyle{plain}
\newtheorem{theorem}{Theorem}
\newtheorem{corollary}{Corollary}
\newtheorem{lemma}{Lemma}

\newtheorem{fact}{Fact}

\theoremstyle{definition}
\newtheorem{definition}{Definition}
\newtheorem{remark}{Remark}

\usepackage[most]{tcolorbox}
\tcbset{
    examplebox/.style={
        colback=pink!15,    
        colframe=pink!30, 
        fonttitle=\bfseries,
        coltitle=black,
        boxrule=1pt,
        arc=1mm,
        left=6pt, right=6pt, top=6pt, bottom=6pt,
        breakable = true,
    }
}

\newcounter{example}
\newenvironment{example}[1][]{
  \refstepcounter{example}%
  \begin{tcolorbox}[examplebox,
    title={Example~\theexample: #1}
  ]
}{
  \end{tcolorbox}
}

\DeclareMathOperator{\RNG}{RNG}
\DeclareMathOperator{\conv}{conv}

\DeclareMathOperator{\dist}{Dist}

\definecolor{gray}{HTML}{f0f0f0}
\colorlet{shadecolor}{gray}

\newcommand{\Smin}[1][]{S_{{\rm min}\ifthenelse{\equal{#1}{}}{}{, #1}}}
\newcommand{\Fmin}{F_{\rm min}}
\newcommand{\Smax}[1][]{S_{{\rm max}\ifthenelse{\equal{#1}{}}{}{, #1}}}
\newcommand{\Fmax}{F_{\rm max}}

\begin{document}
\title{Manipulating heterogeneous quantum resources over a network}
\author{Ray Ganardi}
\email{ray@ganardi.xyz}
\affiliation{School of Physical and Mathematical Sciences, Nanyang Technological
University, 21 Nanyang Link, 637371 Singapore, Republic of Singapore}

\author{Jeongrak Son}
\affiliation{School of Physical and Mathematical Sciences, Nanyang Technological
University, 21 Nanyang Link, 637371 Singapore, Republic of Singapore}

\author{Jakub Czartowski}
\affiliation{School of Physical and Mathematical Sciences, Nanyang Technological
University, 21 Nanyang Link, 637371 Singapore, Republic of Singapore}
\affiliation{School of Physics, Trinity College Dublin, Dublin 2, Ireland}

\author{Seok Hyung Lie}
\email{seokhyung@unist.ac.kr}
\affiliation{Department of Physics, Ulsan National Institute of Science and Technology (UNIST), Ulsan 44919, Republic of Korea}

\author{Nelly H.Y. Ng}
\email{nelly.ng@ntu.edu.sg}
\affiliation{School of Physical and Mathematical Sciences, Nanyang Technological
University, 21 Nanyang Link, 637371 Singapore, Republic of Singapore}
\affiliation{Centre for Quantum Technologies, Nanyang Technological University, 50 Nanyang Avenue, 639798 Singapore}
\date{\today}

\begin{abstract}
Quantum information processing relies on a variety of resources, including entanglement, coherence, non-Gaussianity, and magic.
In realistic settings, protocols run on networks of parties with \emph{heterogeneous} local resource constraints, so different resources coexist and interact. 
Yet, resource theories have mostly treated each resource in isolation, and a general theory for manipulation in such distributed settings has been lacking.
We develop a unified framework for composite quantum resource theories that describes distributed networks of locally constrained parties. 
We formulate natural axioms a composite theory should satisfy to respect the local structure, and from these axioms derive fundamental bounds on resource manipulation that hold universally, independent of the particular network characteristics.
We apply our results to central operational tasks, including resource conversion and assisted distillation, and introduce new methods to construct new resource monotones from this setup.
Our framework further reveals previously unexplored phenomena in the remote certification of quantum resources.
Together, these results establish foundational laws for distributed quantum resource manipulation across diverse physical platforms.
\end{abstract}

\maketitle

\section{Introduction}

What enables quantum information processing protocols to outperform its classical counterpart?
Over the past decades it has become clear there is no single answer: quantum advantage in different tasks relies on different physical resources. 
Some protocols require entanglement for enhancement~\cite{Bennett1992DenseCoding,Bennett1993Teleportation,Jozsa2003EntanglementAlgorithms,Pezze_2009}, while others hinge on coherence~\cite{Hillery2016Coherence,Anand2016CoherenceEntanglement,Ahnefeld2022CoherenceShor}, non-Gaussianity~\cite{Mari_2012,Bartlett_2002}, or magic~\cite{Gottesman_1998,Nielsen_Chuang_2010}. 
This diversity reflects the richness of quantum resources, but it also raises a structural question: how should distinct quantum resources be compared, combined, or traded against one another when they coexist within the same information-processing ecosystem?

Quantum resource theories provide a \textit{lingua franca} for characterizing and manipulating resources~\cite{Chitambar_2019}. 
In this framework, one specifies a set of free states $S$ and free operations $F$ satisfying $F(S)\subseteq S$, while designating those that cannot be generated for free as resource states. 
This conceptually clean abstraction has led to deep insights across a wide range of settings, including entanglement distillation~\cite{Bennett_1996a}, generalized second laws of thermodynamics~\cite{Horodecki_2013,Brandao_2015}, quantitative resource measures~\cite{Vedral_1997,Takagi2019_convex,Regula2018_convex}, and optimal quantum compilation~\cite{Howard_2017}. 
Yet, most resource-theoretic treatments assume \emph{homogeneous} constraints: they focus on a single resource and do not address how multiple, local limitations compose when states and operations are distributed across parties.

The emergence of large-scale quantum networks (also called the quantum internet~\cite{wehner2018QInternet}) urges us to address this gap.
Such networks will enable distributed quantum information-processing tasks such as quantum cryptography, distributed quantum sensing and computing.
However, their nodes may operate on distinct physical platforms and under different experimental conditions, and thus face \emph{heterogeneous} local resource constraints. 
This raises a fundamental question: what laws govern resource manipulation when each participant is subject to a different resource theory? 
More broadly, how do local constraints shape what is globally achievable? 

Two main streams of prior work have sought to address this challenge. 
Early approaches focused primarily on extending entanglement theory, progressively incorporating additional restrictions on local parties to model resource limitations~\cite{Horodecki_2003b,Smolin_2005,Chitambar_2016a,Streltsov_2017,Morris_2019,Dutil_2011,Lami2020_AssistedGaussian,Fiurasek_2002,Eisert_2002,Giedke_2002,Chitambar_2016b,arxiv_Bistron_2024,Ganardi_2025a,Andi_Gu_2025,Synak-Radtke_2005,arxiv_Horodecki_2005a}. 
A complementary approach adopts an axiomatic perspective, combining independent resource-theoretic structures under general principles~\cite{Sparaciari_2017,Sparaciari_2020,WS2024}. 
The axiomatic approach is very general but often too weakly constrained to yield strong operational conclusions.

In this work, we build on the strengths of the axiomatic approach, while introducing physically motivated assumptions tailored to heterogeneous quantum networks. 
We construct a framework that hybridizes heterogeneous local resource theories, under the simplifying assumption that each party is governed by a single local resource constraint. 
Specifically, we define axioms for composite, network-level operations and derive universal bounds on multi-resource manipulation in both single-shot and asymptotic regimes. 
Our framework also provides new resource monotones and reveals operational tasks unique to the heterogeneous setting, including remote certification of resources and resource conversion. 
Together, these results establish a unified, predictive theory of resource manipulation for heterogeneous quantum networks.

\begin{figure*}
	\includegraphics[width=\textwidth]{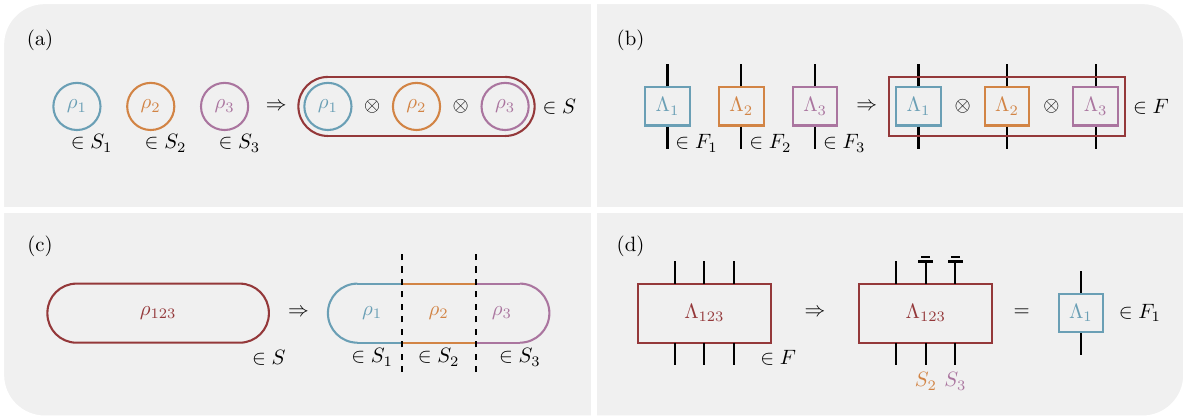}
	\caption{\label{figure:compatibility-rules}
		Rules of compatibility between the local resource theories and the composite resource theory (see \cref{definition: composite RTs} for the formal statement):
		(a) tensor product of locally free states are free,
		(b) tensor product of locally free operations are free,
		(c) all reduced states of a composite free state are locally free, and
		(d) for any free input, all of the effective reduced operations of a composite free operation are locally free. Heuristically, rules (a) and (b) can be thought as product rules, while rules (c) and (d) as marginal rules.
	}
\end{figure*}
\section{Composite resource theory}\label{sec:composing}

Given a Hilbert space $\mathcal{H}$, let $\mathcal{S}(\mathcal{H})$ denote the set of density matrices, and $\mathcal{O}(\mathcal{H})$ denote the set of quantum channels.
We first state a minimal requirement for what constitutes a quantum resource theory.

\begin{definition}\label{def:RT}
	Let $S \subseteq \mathcal{S}(\mathcal{H}), F \subseteq \mathcal{O}(\mathcal{H})$.
	We say that $(S, F)$ is a \textit{resource theory} if $F$ includes the identity channel, is closed under concatenation, and $F(S) \subseteq S$.
	We call $S$ the set of free states and $F$ the set of free operations.
	The states in $\mathcal{S}(\mathcal{H}) \setminus S$ are called resource states.
\end{definition}

A resource theory is convex when $S$ and $F$ are convex. 
Closure under concatenation expresses the natural requirement that a sequential composition of free operations remains free. While this property is not universally imposed, for example in resource theories where time complexity is a resource, it holds in the vast majority of physically motivated settings and will be assumed throughout this work.
Given a set of free states $S$, the associated set of \emph{resource non-generating} (RNG) operations is defined as
\begin{equation}
	\RNG(S) = \Bqty{\Lambda \in \mathcal{O}(\mathcal{H}) \,|\, \Lambda(S) \subseteq S}.
\end{equation}
The pair $(S, \RNG(S))$ defines the RNG resource theory associated with $S$. Prominent examples include unital channels when $S$ consists of the maximally mixed state, and Gibbs-preserving maps in quantum thermodynamics, where $S$ is the Gibbs state~\cite{Faist2015}.

We now consider a composite system consisting of multiple parties, where each party is governed by a local resource theory $(S^{(i)}, F^{(i)})$.
A canonical example is a quantum network with multiple nodes, with each node being constrained by a specific resource and the corresponding local resource theory.
In general, different nodes may be governed by different resource theories. 
This naturally raises the question: which global (i.e., composite) resource theories are compatible with a given collection of local ones?
To answer this question, we introduce minimal consistency conditions for resource networks (illustrated in \cref{figure:compatibility-rules}).

\begin{definition}\label{definition: composite RTs}
  Let $(S^{(i)}, F^{(i)})$ be a collection of local resource theories on $\mathcal{H}^{(i)}$.
  We say that $(S, F)$ is a \emph{composite} resource theory on $\bigotimes_i \mathcal{H}^{(i)}$ if $(S, F)$ is a resource theory, and additionally the following conditions are satisfied: \vspace{-0.2cm}
  \begin{enumerate}[label=\alph*),itemsep=0pt,leftmargin=*]
	  \item\label{condition:free product states} \emph{free product states}: if $\rho^{(i)} \in S^{(i)}$, then $\bigotimes_i \rho^{(i)} \in S$.
	  \item\label{condition:free product operations} \emph{free product operations}: if $\Lambda^{(i)} \in F^{(i)}$, then $\bigotimes_i \Lambda^{(i)} \in F$.
	  \item\label{condition:free marginal states} \emph{free marginal states}: for any $\rho \in S$, we have $\Tr_{\overline{i}} \rho \in S^{(i)}$, where $\Tr_{\overline{i}} $ denotes the partial trace on all subsystems except the $i$-th one.
	  \item\label{condition:free marginal operations} \emph{free marginal operations}: for any $\Lambda \in F$ and $\rho^{(j)} \in S^{(j)}$, the operation $X^{(i)} \mapsto \Tr_{\overline{i}} \Lambda(X^{(i)} \otimes \bigotimes_{j \neq i} \rho^{(j)}) \in F^{(i)}$ for all $i$.
  \end{enumerate}
\end{definition}

The first two conditions, providing interior bounds, ensure that each party can control its own resource locally and independently.
The latter two ensure that the composite theory reduces to local ones when restricted to local systems, thus providing exterior bounds.
In particular, \ref{condition:free marginal states} and \ref{condition:free marginal operations} exclude pathological scenarios in which a local resource could be generated by preparing a composite free state and discarding subsystems. 
We compare our framework in detail with earlier compositional strategies in \ref{supp:related-works}.

\cref{definition: composite RTs} is reminiscent of the famous Brandão-Plenio axioms~\cite{Brandao_2010} (see also \cref{def:BPaxioms} in \ref{section:homogeneous-composition}), which we use as a guiding principle.
These axioms guarantee that the resource theory is well-behaved when considering multi-copy transformations.
Conditions \ref{condition:free product states} and \ref{condition:free marginal states} concerning free states are identical to the Brandão-Plenio axioms (3) and (4) in \cref{def:BPaxioms}.
The other two concerning composite free operations are novel, but are automatically satisfied when preparation of any free state and partial trace are allowed.
Thus \cref{definition: composite RTs} can be regarded as a minimal generalization of the Brandão-Plenio axioms that accommodates heterogeneous resources.
Note that \cref{definition: composite RTs} is deliberately minimal: it does not assume partial trace is free in general. 
Alternative axiomatizations are discussed in \ref{supp:heterogeneous}.

Finally, \cref{definition: composite RTs} does not uniquely specify a composite resource theory. 
Even in the homogeneous setting, multipartite extensions of standard resource theories already reveal residual freedom in how local structures can be composed. 
This ambiguity becomes more pronounced in heterogeneous scenarios, where local theories differ across parties.
Nevertheless, we show that it is possible to define well-motivated \emph{minimal} and \emph{maximal} composite constructions, that act as central benchmarks throughout this work. 

\begin{definition}\label{definition: min max sets}
  Let $(S^{(i)}, F^{(i)})$ be a collection of local resource theories. 
  We define \emph{minimal composite free states} as $\Smin = \conv \Bqty{\bigotimes_i \rho^{(i)} \,|\, \rho^{(i)} \in S^{(i)}}$ and \emph{minimal composite free operations} as $\Fmin = \conv \Bqty{\bigotimes_i \Lambda^{(i)} \,|\, \Lambda^{(i)} \in F^{(i)}}$.
  We refer to $(\Smin, \Fmin)$ as the minimal composite theory. 
  Similarly, we define the \emph{maximal composite free states} as $\Smax = \Bqty{\rho \,|\, \Tr_{\overline{i}} \rho \in S^{(i)}}$.
\end{definition}

One can verify that $(\Smin, \Fmin)$ and $(\Smax, \Fmin)$ are valid composite theories satisfying \cref{definition: composite RTs}.
Moreover, Definitions~\ref{definition: composite RTs} and~\ref{definition: min max sets} immediately imply the following:
\begin{fact}\label{lemma:composite-free-states}
  For any convex composite resource theory $(S, F)$, we have $\Smin \subseteq S \subseteq \Smax$ and $\Fmin \subseteq F$.
\end{fact}

The extremal composite free state sets $\Smin,\Smax$ were identified in Ref.~\cite{Son_RobustCat} in a context related but restricted to homogeneous compositions.
In particular, for some resources like coherence and athermality, $\Smin$ is exactly the set of homogeneous composite free states.
$\Fmin$ is first defined in this work from our minimal axioms. 

One can easily identify previously studied resource theories as special cases of this general framework; see \ref{supp:composition} for comprehensive discussion.
Here we consider the case where no local restrictions exist and illustrate how different variants of entanglement theory emerge from the composition structure. 
Suppose that $S^{(i)} = \mathcal{S}(\mathcal{H}^{(i)})$ and $F^{(i)} = \mathcal{O}(\mathcal{H}^{(i)})$.
Then $\Smin$ corresponds to the set of separable states, and $\Fmin$ is the well-known set of local operations and shared randomness (LOSR)~\cite{arxiv_Dukaric_2008,Buscemi_2012a,Forster_2009,Schmid2020typeindependent,Schmid_2023}.
The resulting minimal composite theory $(\Smin,\Fmin)$ is thus precisely the resource theory of entanglement under LOSR operations. 
A closely related and widely studied alternative is the set of local operations and classical communication (LOCC)~\cite{Bennett_1996b,Donald_2002,Chitambar_2014}, which models the scenario in which quantum parties are connected by a classical communication network. 
Motivated by the central role of LOCC in quantum information theory, this notion has been generalized to \emph{locally free operations and classical communication} (LFOCC)~\cite{Horodecki_2003b,Synak-Radtke_2005,arxiv_Horodecki_2005a,Fiurasek_2002,Eisert_2002,Giedke_2002,Chitambar_2016b,Streltsov_2017,arxiv_Bistron_2024,Ganardi_2025a,Andi_Gu_2025}, formally defined in \ref{supp:composition}, which captures the setting in which each party is subject to local resource constraints while having access to classical communication. 
Finally, the set $\RNG(\Smin)$ corresponds to the class of non-entangling operations.

By contrast, the maximal composite free states $\Smax$ now coincides with the set of all quantum states, and the associated theory $(\Smax,\RNG(\Smax))$ becomes trivial, in the sense that all states are free and all operations are allowed. In this special case, a maximal set of composite free operations emerges, reflecting the absence of any local resource constraints. However, we will show soon that such a maximal composite set is generally not to be expected. 

\section{Results}
A central task in any resource-theoretic framework is to characterize the possibility of state transformations. Given a resource theory $(S,F)$, this amounts to determining whether a transformation $\rho \xrightarrow{(S, F)} \sigma$ is possible, i.e., whether there exists a free operation that maps $\rho$ to $\sigma$. A standard tool for addressing such questions is provided by \emph{resource monotones}: real-valued functions that quantify the amount of resource contained in a state and are monotonically non-increasing under free operations. A canonical example is the \emph{relative entropy of a resource}, defined as
\begin{align}
	D(\rho \| S)
	&=
	\inf_{\mu \in S} D(\rho \| \mu),
\end{align}
where $D(\rho \| \sigma)$ denotes the quantum relative entropy.

Beyond the single-shot setting, one is often interested in asymptotic state transformations involving many copies of a resource state. In this regime, the central quantity of interest is the achievable \emph{conversion rate}. We say that $\rho$ can be transformed into $\sigma$ at rate $\frac{m}{n}$ if
$\rho^{\otimes n} \xrightarrow{(S, F)} \sigma^{\otimes m}$. 
The maximal achievable rate is denoted by $R(\rho \xrightarrow{(S,F)} \sigma)$.
Moreover, we typically allow non-exact transformations in the asymptotic setting; we provide a precise formulation in \ref{supp:preliminaries}.

Our results concern four central aspects of resource manipulation: \vspace{-0.2cm}
\begin{enumerate}[label=\textbf{\Alph*.}, leftmargin=*,itemsep=-2pt]
	\item Existence and properties of extremal composite theories.
	\item Conditions for single-shot state transformations.
	\item Universal upper bounds on asymptotic conversion rates.
	\item Operational advantages unique to composite settings, including remote resource certification.
\end{enumerate}
Each of these aspects is developed in the subsections below.
Importantly, all of our results in this section pertain to convex resource theories.

\subsection{Extremal composite theories}\label{subsec:extremal}

While \cref{lemma:composite-free-states} shows that the set of composite free states is always bounded between two extremal constructions, an analogous maximal set of composite free operations need not exist. More precisely, there is generally no set $\Fmax$ such that, for every composite resource theory $(S,F)$ satisfying \cref{definition: composite RTs}, one has $F \subseteq \Fmax$ (see \ref{supp:no-fmax} for proof and further discussion).
This is a consequence of the fact that $\RNG$s do not behave monotonically with respect to the set of free states, as below.
\begin{lemma}\label{lem:RNG_min_max}
	$S \subseteq S'$ does not imply $\RNG(S) \subseteq \RNG(S')$ nor $\RNG(S') \subseteq \RNG(S)$.
\end{lemma}
\begin{proof}
	A simple single-qubit example suffices. 
	To see that $S \subseteq S'$ does not imply $\RNG(S) \subseteq \RNG(S')$, take $S = \Bqty{\mathbf{1}/2}, S' = \Bqty{\mathbf{1}/2, \ketbra{0}}$, and note that Pauli $\sigma_X$ is in $\RNG(S)$ but not $\RNG(S')$.
	For the second statement, take $S = \Bqty{\ketbra{0}}$ and $S'$ to be the set of all density matrices.
	Then any CPTP map is in $\RNG(S')$, including those that do not preserve $\RNG(S)$.
\end{proof}
This prevents the set of composite free operations from inheriting the extremal bounds on composite free states.
In particular, we find an example where the union of $\RNG(\Smin)$ and $\RNG(\Smax)$ is no longer closed under concatenation.
Moreover, one can also show that in this example, if the free operation set $F$ contains both $\RNG(\Smin)$ and $\RNG(\Smax)$, then any set of states $S$ satisfying $\Smin \subseteq S \subseteq \Smax$ is not preserved by $F$. 

The absence of such a maximal set of free operations obstructs a straightforward derivation of universal resource-theoretic bounds that are independent of the specific composite theory. 
Indeed, if a maximal set $\Fmax$ did exist, then the optimal performance achievable under $\Fmax$ would constitute an ultimate bound on all processes compatible with the given local restrictions since any $\Fmax$-monotone would necessarily be monotonic under every admissible composite resource theory.
This would have mirrored the standard situation in single-resource theories: the performance of a resource theory $(S,F)$ in any given task is upper bounded by that of  $(S,\RNG(S))$.

This issue is particularly poignant in a multi-resource setting such as ours, where it is \emph{a priori} unclear which resource-theoretic structures should be admitted. As an illustrative example, consider a multipartite theory of quantum thermodynamics in which different subsystems are held at different temperatures. For such a theory to be well behaved, one expects that the effective operation on each party reduces to local thermal operations, whenever all other parties are initialized in their respective Gibbs states---this is precisely the requirement captured by the free marginal operations condition \ref{condition:free marginal operations}.
What, then, constitutes the correct set of composite free operations?
Previous works~\cite{Bera_2021,arxiv_Bistron_2024} approached this by selecting particular composite free operation sets, but it is unclear whether there is a universal second law of thermodynamics that is applicable regardless of the particular choice of a composite theory.
Our framework sidesteps this issue and allows a derivation of laws that apply universally, i.e. for any composite theory.

\subsection{Single-shot transformations}
Despite the absence of $\Fmax$, we show that it is still possible to derive universal resource manipulation laws for heterogeneous resources.
We start with a necessary condition on single-shot transformations that depends only on the local structure, irrespective of the composition. 

\begin{theorem}~\label{theorem:single-shot}
  Let $(S, F)$ be a composite resource theory.
  If $\rho \xrightarrow{(S, F)} \sigma$, then $D(\rho \| \Smin) \geq D(\sigma \| \Smax)$.
\end{theorem}
\begin{proof}
  By \cref{lemma:composite-free-states}, $\Smin \subseteq S \subseteq \Smax$, thus for any $\rho$, 
  \begin{align}
    D(\rho \| \Smax)
    \leq
    D(\rho \| S)
    \leq
    D(\rho \| \Smin).
  \end{align}
  Combining this with the monotonicity of $D(\rho \| S)$, i.e. $\rho \xrightarrow{(S,F)} \sigma$ implies $D(\rho \| S) \geq D(\sigma \| S)$, we obtain the claim.  
  Alternatively, we can use conditions~\ref{condition:free marginal states} and \ref{condition:free marginal operations} in \cref{definition: composite RTs} to show that for any $\Lambda \in F$, we have $\Lambda(\Smin) \subseteq \Smax$.
\end{proof}

This theorem may appear redundant as there exists an obviously tighter inequality $D(\rho \| S) \geq D(\sigma \| S)$ for any given composite theory $(S, F)$. 
However, the quantities $D(\rho \| S)$ and $D(\sigma \| S)$ cannot be evaluated using knowledge of the local free state sets $S^{(i)}$ alone, unless the composite free state set $S$ is explicitly specified. Even when $S$ is known, the associated optimization problem may be computationally complex.
By contrast, the quantities $D(\rho \| \Smin)$ and $D(\sigma \| \Smax)$ depend only on the extremal composite free state sets $\Smin$ and $\Smax$, which are fully determined by the local sets. Moreover, optimization over these extremal sets is typically simpler than optimization over an arbitrary composite set $S$.
For example, when we have a network composed of $n$ qubits and $S$ is the set of stabilizer states, $\Smax$ is characterized by $8n$ linear inequalities whereas $S$ is characterized by an intricate condition on the stabilizing group of its extreme points.
In this sense, \cref{theorem:single-shot} provides a practically useful necessary condition for single-shot state transformations in heterogeneous settings.

Beyond its practical utility, \cref{theorem:single-shot} demonstrates the existence of resource manipulation laws that are independent of the choice of composite resource theory. Notably, the only property of the quantum relative entropy used in the proof is the data-processing inequality. 
As a result, we can substitute relative entropy with any divergence satisfying data-processing inequality to obtain an alternative universal necessary conditions for state transformations.

We emphasize that the quantities evaluated on the initial and final states, $D(\rho \| \Smin)$ and $D(\sigma \| \Smax)$, are \emph{not} the same function. 
Consequently, they cannot be interpreted as universal resource monotones. 
Indeed, \cref{example:no-fmax} in \ref{supp:no-fmax} shows that $D(\rho \| \Smin)$ fails to be monotonic for the composite theory $(\Smax, \RNG(\Smax))$, while $D(\rho \| \Smax)$ fails to be monotonic for $(\Smin, \RNG(\Smin))$. Nevertheless, as we show next, these quantities reduce to standard resource monotones in the special case where the state is uncorrelated and only a single subsystem carries the resource.

\begin{lemma}\label{lemma:uncorrelated}
	Let $(S, F)$ be a composite resource theory and $\rho = \bigotimes_{i} \rho^{(i)}$ be an uncorrelated state.
	Suppose that for all $i > 1$, $\rho^{(i)} \in S^{(i)}$ is a free state.
	Then
\begin{equation}
	D(\rho\|S) = D(\rho^{(1)} \| S^{(1)}).
\end{equation}
\end{lemma}
\begin{proof}
    Observe that
\begin{align*}\label{eq:single_vs_composite}
	D\pqty{\rho \| \Smin}
	\leq
	D\pqty{\rho^{(1)} \Big\| S^{(1)}}
	&\leq
	D\pqty{\rho \| \Smax}
   \leq
	D\pqty{\rho \| \Smin}.
\end{align*}
The first inequality comes from considering the state $\mu^{(1)} \in S^{(1)}$ that achieves $D(\rho^{(1)}\|\mu^{(1)}) = D(\rho^{(1)}\|S^{(1)})$ and combining with $ D(\rho^{(1)} \| \mu^{(1)}) = D\pqty{ \bigotimes_i \rho^{(i)} \| \mu^{(1)} \otimes \bigotimes_{i>1} \rho^{(i)}} \geq D(\rho\|\Smin)$.
The second inequality comes from data-processing inequality under partial trace $\Tr_{\overline{1}}$ and the third from $\Smin\subseteq\Smax$.
Essentially, we have shown $D(\rho \| \Smin) = D(\rho \| \Smax) = D(\rho^{(1)} \| S^{(1)})$, and the claim follows by noting $\Smin\subseteq S\subseteq \Smax$ from \cref{lemma:composite-free-states}.
\end{proof}

Now, consider the following genuinely heterogeneous resource manipulation task: converting one resource into another.
Alice holds a resource and wishes to convert it into another resource in Bob's system for his use.
The performance of this conversion is bounded by the following corollary of \cref{theorem:single-shot} and \cref{lemma:uncorrelated}.

\begin{corollary}~\label{corollary:conversion-single-shot}
  Let $(S, F)$ be a composite resource theory and $\mu^{(1)} \in S^{(1)}, \mu^{(2)} \in S^{(2)}$ be free states.
  If $\rho^{(1)} \otimes \mu^{(2)} \xrightarrow{(S, F)} \mu^{(1)} \otimes \rho^{(2)}$, then $D(\rho^{(1)} \| S^{(1)}) \geq D(\rho^{(2)} \| S^{(2)})$.
\end{corollary}\vspace{0.01cm}

\begin{example}[Conversion between coherence and entanglement]\label{ex:entanglement_coherence_conversion}
    Suppose that parties $1$ and $2$ are governed by coherence and entanglement theories, respectively. 
    For simplicity, assume that $1$ is a one-qubit system and $2 = AB$ is a two-qubit system with a partition $A|B$.
    $S^{(1)}$ is then the set of incoherent qubit states and $S^{(2)}$ is the set of separable two-qubit states.\\
    
    We first consider the conversion from coherence to entanglement, and show that \cref{corollary:conversion-single-shot} gives a tight bound. 
    This was first observed in Ref.~\cite{Streltsov_2015} in a slightly more general setting.
    A maximally resourceful state for party $1$ is $\ket{+}^{(1)} = \frac{1}{\sqrt{2}}(\ket{0}^{(1)}+\ket{1}^{(1)})$ whose relative entropy of coherence $D(\dyad{+}^{(1)} \| S^{(1)}) = 1$.
    Similarly, a maximally resourceful state for party $2$ is $\ket{\Phi^{+}}^{(AB)} = \frac{1}{\sqrt{2}}(\ket{00}^{(AB)}+\ket{11}^{(AB)})$, and its relative entropy of entanglement is also $D(\Phi^{+(AB)} \| S^{(AB)}) = 1$.
    \cref{corollary:conversion-single-shot} implies that at most one copy of $\ket{\Phi^{+}}^{(AB)}$ can be obtained from $\ket{+}^{(1)}$.
    We show that this bound is saturated by constructing a composite theory that allows the single-shot transformation $\ket{+}^{(1)}\otimes\ket{00}^{(AB)} \to \ket{0}^{(1)}\otimes\ket{\Phi^{+}}^{(AB)}$.
    
    Let us start with an arbitrary set of composite free states $S^{(12)}$.
    Next, we define the following CPTP map
    \begin{align*}
        \Lambda^{(12)}(\rho^{(12)})
        &=
        \dyad{0}^{(1)}\otimes\mathcal{U}\left(\Tr_{2}[\rho^{(12)}]\otimes \dyad{0}^{(B)}\right),
    \end{align*}
    where $\mathcal{U}$ relabels $1$ to $A$ and then applies CNOT gate to $AB$ with $A$ as the control.
    Now, take an arbitrary $\mu^{(12)} \in S^{(12)}$.
    Due to condition~\ref{condition:free marginal states} in \cref{definition: composite RTs}, $\Tr_2 \bqty{\mu^{(12)}}$ must be incoherent.
    Furthermore, CNOT gate cannot entangle two incoherent states and thus $\mathcal{U}(\Tr_{2}[\mu^{(12)}]\otimes \dyad{0}^{(B)})$ is a separable state.
    Therefore $\Lambda^{(12)}(\mu^{(12)})$ is a product of two free states, and in particular $\Lambda^{(12)}(\mu^{(12)}) \in S^{(12)}$.
    This implies that $\Lambda^{(12)}$ is an allowed operation for the composite theory $(S^{(12)}, \RNG(S^{(12)}))$.
    We also verify that $\Lambda^{(12)}$ achieves the desired transformation $\ket{+}^{(1)}\otimes\ket{00}^{(AB)} \to \ket{0}^{(1)}\otimes\ket{\Phi^{+}}^{(AB)}$, showing the claim. \\

	Interestingly, we find that the reverse transformation $\ket{0}^{(1)}\otimes\ket{\Phi^{+}}^{(AB)}\to \ket{+}^{(1)}\otimes\ket{00}^{(AB)}$ is forbidden. 
	In fact, it is not possible to convert any entangled state into some coherent state in any composite theory.
	To see this, let us fix an arbitrary $S^{(12)}$.
        Recall if $\mu^{(1)}\in S^{(1)}$ and $\mu^{(2)}\in S^{(2)}$, then any $\Lambda^{(12)}\in\RNG(S^{(12)})$ must transform $\mu^{(1)}\otimes\mu^{(2)}$ into $\mu'^{(12)}$ with the marginals $\mu'^{(1)}\in S^{(1)}$ and $\mu'^{(2)}\in S^{(2)}$.
	However, any (entangled) state $\rho^{(2)}$ can be written as an affine combination of separable states in $S^{(2)}$~\cite{Zyczkowski_1998}. 
	Hence, for any $\rho^{(2)}$, we can find real coefficients $q_{i}$ that sum up to one such that $\Tr_{2}[\Lambda^{(12)}(\mu^{(1)}\otimes\rho^{(2)})] = \sum_{i} q_{i}\mu_i^{(1)}$, and $\mu_i^{(1)}\in S^{(1)}$ are incoherent states. 
	Since any affine combination of incoherent states is also incoherent, we have shown that $\ket{0}^{(1)}\otimes\ket{\Phi^{+}}^{(AB)} \not\to \ket{+}^{(1)}\otimes\ket{00}^{(AB)}$ in any composite theory, i.e. the conversion from entanglement to coherence is impossible. 
\end{example}

The resource-conversion task illustrates a broader principle: a monotone in one resource theory can be turned into a monotone in another in a systematic way.
Special cases of this idea are known: the correspondence between coherence and entanglement for maximally correlated states~\cite{Streltsov_2015,Winter_2016a,Chitambar_2016d} allows construction of a coherence monotone from any entanglement monotone~\cite{Streltsov_2015}; analogous constructions exist for other pairs of resource theories where magic can be measured through non‑Gaussianity~\cite{Hahn_2025}, total correlation through non‑Gaussianity~\cite{arxiv_Kaifeng_Bu_2025}, or non‑Markovianity through entanglement~\cite{Kolondynski_2020}.
Our framework elevates this observation to a general statement valid for \emph{arbitrary} pairs of resource theories.
Furthermore, the induced monotone inherits desirable properties such as continuity and faithfulness from the parent monotone under mild assumptions. 
We refer readers to \ref{supp:monotone} for proof and additional discussions.

\begin{lemma}~\label{proposition:monotone}
  Let $R_2$ be an $(S^{(2)}, F^{(2)})$-monotone.
  Then given any composite theory $(S, F)$,
  \begin{equation}\label{eq:2monotone}
      R_1 (\rho^{(1)}) := \sup
      R_2\pqty{\Tr_1 \bqty{\Lambda(\rho^{(1)} \otimes \mu^{(2)})}}
  \end{equation}
  is an $(S^{(1)}, F^{(1)})$-monotone, where the supremum is taken over all  $\Lambda \in F$, and $\mu^{(2)} \in S^{(2)}$. 
\end{lemma}

\cref{proposition:monotone} implies that the task of measuring one resource from another can be constructed for any composite theory $(S,F)$, albeit the performance of the monotone would depend on the composition.
The tightest one can be constructed by additionally optimizing Eq.~\eqref{eq:2monotone} over all composite theories $(S,F)$.

\subsection{Asymptotic transformations}\label{section:asymptotic}
\subsubsection{Multi-copy vs multi-partite setting}
To address asymptotic transformations, we must extend to multiple copies of a system.
For a homogeneous resource theory, a natural multi‑copy extension exists: given $(S,F)$ with $S\subseteq\mathcal{S}(\mathcal{H})$ and $F\subseteq\mathcal{O}(\mathcal{H})$, there is a family of free state sets $S_n \subseteq \mathcal{S}(\mathcal{H}^{\otimes n})$ and free operations $F_n \subseteq \mathcal{O}(\mathcal{H}^{\otimes n})$. 
The resulting resource theory is then described by the pair $(\bigcup_{n=1}^\infty S_n, \bigcup_{n=1}^\infty F_n)$. 
The Brandão-Plenio axioms ensure that such extensions are well behaved, as discussed in \ref{section:homogeneous-composition}. 
We thus ask whether a composite resource theory with heterogeneous local resources automatically inherits the Brandão-Plenio axioms when each local theory does.

To formulate this question precisely, consider a composite system $\mathcal{H} = \bigotimes_i \mathcal{H}^{(i)}$, with each system $\mathcal{H}^{(i)}$ governed by a local resource theory structure $(\bigcup_{n=1}^{\infty} {S^{(i)}_n}, \bigcup_{n=1}^{\infty}{F^{(i)}_n})$.
For each $n$, we choose a composite theory $(S_n, F_n)$ compatible with the local theory $(S^{(i)}_n, F^{(i)}_n)$ following \cref{definition: composite RTs}. 
Then, $(\bigcup_{n=1}^{\infty} S_n, \bigcup_{n=1}^{\infty} F_n)$ defines a composite theory that can handle multi-copy transformations.
For brevity, we assume that operations changing the number of copies, such as partial trace, are included in this notation.
We show that even if all local extensions $(S^{(i)}_n, F^{(i)}_n)$ satisfy the Brandão-Plenio axioms, a composite theory $(S_n, F_n)$ may not.

\begin{example}[Composition does not preserve Brandão-Plenio axioms]
    Consider parties 1 and 2, each restricted by local resource theories with families of free states $\{S_n^{(1)}\}, \{S_n^{(2)}\}$ satisfying Brandão-Plenio axioms.
    In particular, $S_m^{i1)} \otimes S_n^{(i)} \subseteq S_{m+n}^{(i)}$ for both $i =1,2$.
    First, let $\Smax^{(12)},\Smin^{(12)}$ denote the maximal and minimal composite free state sets on two parties for $n=1$. 
    Next, for $n=2$, we choose composite free states to be
    \begin{equation}
        \hat S^{(12)} = {\rm conv} \left\lbrace \Smax^{(12)} \otimes \Smin^{(12)} \right\rbrace.
    \end{equation}
    In other words, we compose two parties by taking 1) the maximal composition on the first copy of each party, 2) the minimal composition on the second copy, and 3) allowing for convex combinations. 
    Note that this is a valid composite resource theory according to \cref{definition: composite RTs}.

    Now we can choose a composite theory with the following free state sets $S_{1} = \Smax^{(12)}$ for $n = 1$, and $S_{2} = \hat S^{(12)} $ for $n = 2$.
    In this composite theory, $S_{1} \otimes S_{1} \not\subseteq S_{2}$ because for any state $\mu\in S_1$ but $\mu\notin \Smin^{(12)}$, we have $\mu^{\otimes 2} \notin S_2$.
\end{example}

This example confirms that the compatibility conditions in \cref{definition: composite RTs} are logically distinct from the Brandão-Plenio axioms. 
The reason is conceptual: the Brandão-Plenio axioms concern multiple \emph{copies} of a system, while \cref{definition: composite RTs} concerns one copy of a system consisting multiple \emph{parties}. 
Both scenarios involve tensor products mathematically, but they represent different physical notions of composing systems. Motivated by this distinction, we restrict our attention to composite theories $({S_n}, {F_n})$ that satisfy the Brandão-Plenio axioms, ensuring that asymptotic multi-copy transformations are well behaved. 
Crucially, the extremal composite constructions introduced earlier are compatible with this requirement:
suppose for each $n$ we take $\Smin[n] = \conv\{\bigotimes_i \rho^{(i)} \,|\, \rho^{(i)} \in S^{(i)}_n\}$; 
then, when the local theories $\{S_n^{(i)}\}$ satisfy the Brandão-Plenio axioms, so does the sequence $\{\Smin[n]\}$. 
An analogous statement holds for the maximal construction ${\Smax[n]}$.
For notational simplicity, we henceforth write $\bigcup_{n=1}^\infty {S_n}, \bigcup_{n=1}^\infty{F_n}$ simply as $S, F$.

\subsubsection{Upper bounds on asymptotic composite transformation rates}
We now establish an asymptotic analogue of \cref{theorem:single-shot}, that involves the regularized relative entropy of a resource
\begin{align}
  D^{\infty}(\rho \| S)
  &=
    \lim_{n \to \infty} \frac{1}{n} D(\rho^{\otimes n} \| S_n).
\end{align}
In addition to asymptotic continuity, the regularized relative entropy is weakly additive, i.e. $D^{\infty} (\rho^{\otimes k} \| S) = k D^{\infty} (\rho \| S)$.
Hence, it serves perfectly as a monotone for asymptotic transformations~\cite{Brandao_2008,Brandao_2010a,Horodecki_2002,Donald_2002}.

\begin{theorem}~\label{theorem:asymptotic}
  For any composite resource theory $(S, F)$,
  \begin{align}
  R(\rho \xrightarrow{(S, F)} \sigma)
  &\leq
  \frac{D^{\infty}\pqty{\rho \| \Smin}}{D^{\infty}\pqty{\sigma \| \Smax}}.
  \end{align}
\end{theorem}
\begin{proof}
  We have $R(\rho \xrightarrow{(S, F)} \sigma) \leq \frac{D^{\infty}(\rho \| S)}{D^{\infty}(\sigma \| S)}$ since $(S, F)$ is assumed to satisfy the Brandão-Plenio axioms and thus $D^{\infty}$ is an asymptotically continuous and weakly additive monotone.
  To obtain the claim, it is sufficient to prove
  \begin{align}
    D^{\infty} (\rho \| \Smax)
    \leq
    D^{\infty} (\rho \| S)
    \leq
    D^{\infty} (\rho \| \Smin),
  \end{align}
  which follows from dividing by $n$ and taking $\lim_{n \to \infty}$ in
  \begin{align}
    D(\rho^{\otimes n} \| \Smax[n])
    \leq
    D(\rho^{\otimes n} \| S_n)
    \leq
    D(\rho^{\otimes n} \| \Smin[n]).
  \end{align}
\end{proof}

Again the significance of \cref{theorem:asymptotic} lies in its universality.
In particular, the bound depends only on the local free state families ${S_n^{(i)}}$ through the extremal constructions $\Smin$ and $\Smax$. 
As a result, the asymptotic transformation rate can be bounded without specifying how the local resource theories are exactly composed at the global level.
Moreover, we show that the bound is always non-trivial under mild conditions on the local theories.

\begin{lemma}
  Let $\pqty{S^{(i)}, F^{(i)}}$ be local resource theories whose regularized relative entropies are faithful, i.e. $D^{\infty} (\sigma \| S^{(i)}) = 0$ if and only if $\sigma \in S^{(i)}$.
  Then, $D^{\infty} \pqty{\rho \| \Smax}$ is a faithful monotone.
\end{lemma}
\begin{proof}
	By monotonicity of relative entropy under partial trace, we have
  \begin{align}
    D\pqty{\rho^{\otimes n} \| \Smax[n]}
    &\geq
      \max_i
      D\pqty{\pqty{\Tr_{\overline{i}} \rho}^{\otimes n} \| S^{(i)}}.
  \end{align}
  By definition, $\rho \not\in \Smax$ implies that either $\Tr_{\overline{i}} \rho \not\in S^{(i)}$ for some $i$.
  The claim then follows from the faithfulness of $D^{\infty} (\rho \| S^{(i)})$.
\end{proof}

Unlike the single-shot setting (\cref{theorem:single-shot}) where any divergence satisfying data-processing inequality may  be used, here the regularized relative entropy is singled out.
This is a consequence of the generalized asymptotic equipartition property~\cite{arxiv_Fang_2024a}: suitably smoothed regularizations of many divergences converge to the regularized relative entropy.

Let us examine the implications of \cref{theorem:asymptotic} for networks admitting locally free operations and classical communication (LFOCC). 
Since LFOCC is a subset of LOCC, standard entanglement theory yields the bound $R(\rho \xrightarrow{(\Smin,  LFOCC)} \sigma) \leq D^{\infty}(\rho \| SEP)/D^{\infty}(\sigma \| SEP)$~\cite{Donald_2002}.
However, this only accounts for the limitations due to the correlations contained in $\rho, \sigma$, and is insensitive to the local resources.
This is reflected in that the regularized relative entropy of entanglement of a product state is always zero, regardless of whether it is free or resourceful.
\cref{theorem:asymptotic} provides a complementary bound that solely captures the limitations due to the resources in $\rho, \sigma$, while neglecting the correlation content.
Crucially, the strength of \cref{theorem:asymptotic} is that it applies uniformly to \emph{all} composite theories, not only $(\Smin, LFOCC)$.

Another interesting corollary of \cref{theorem:asymptotic} is the following bound on the asymptotic conversion between different types of local resources.
\begin{corollary}[Resource conversion]~\label{corollary:conversion}
  Let $(S, F)$ be a composite resource theory.
  The resource conversion rate from $\rho^{(1)}$ to $\sigma^{(2)}$ in any composite theory is bounded as $R(\rho^{(1)} \to \sigma^{(2)}) \leq D^{\infty}(\rho^{(1)} \| S^{(1)}) / D^{\infty}(\sigma^{(2)} \| S^{(2)})$.
\end{corollary}

This bound in \cref{corollary:conversion} is not tight for LFOCC composite theory as the rate $R(\rho^{(1)} \to \sigma^{(2)}) = 0$ for any resourceful $\rho^{(1)}, \sigma^{(2)}$.
This is because LFOCC always apply locally free operation to each party, which cannot transform an initially free local state to a non-free one (see \cref{proposition:lfocc}).
However, the bound can be saturated by certain composite theories, for example with the asymptotically resource non-generating (ARNG) construction~\cite{Brandao_2010a,Brandao_2008,Lami_2025,arxiv_Hayashi_2024a}.

\subsubsection{A universal upper bound for assisted distillation}

A particular resource manipulation task over networks has been studied under the name of assisted distillation~\cite{Chitambar_2016a,Streltsov_2017,Morris_2019,Smolin_2005,Dutil_2011,Lami2020_AssistedGaussian}.
In this setting, a state $\rho^{(AB)}$ is shared between the system $B$ and the assistant $A$.
The operations on the system $B$ are limited by the resource, while the assistant $A$ does not have such restrictions.
The goal is to maximize the assisted distillation rate, i.e. the transformation rate from $\rho^{(AB)}$ to a golden unit of resource in $B$ denoted as $\Phi^{(B)}$.
This setting is in fact identical to a heterogeneous composite resource theory where the local theory $(S^{(A)} = \mathcal{S}(\mathcal{H}^{(A)}), F^{(A)} = \mathcal{O}(\mathcal{H}^{(A)}))$ for $A$ is trivial, while $(S^{(B)},F^{(B)})$ for $B$ is not.
We apply \cref{theorem:asymptotic} to derive a universal bound for the assisted distillation rate. 

\begin{corollary}[Assisted distillation]~\label{corollary:assisted-distillation}
  Let $(S, F)$ be a composite resource theory.
  The assisted distillation rate from $\rho^{(AB)}$ is bounded as
  \begin{equation}
      R(\rho^{(AB)} \xrightarrow{(S, F)} \Phi^{(B)}) \leq \frac{D^{\infty}(\rho^{(AB)} \| \Smin)}{D^{\infty}(\Phi^{(B)} \| S^{(B)})}.
  \end{equation}
\end{corollary}
\begin{proof}
  Due to \cref{lemma:uncorrelated}, $D^{\infty}(\mu^{(A)} \otimes \Phi^{(B)} \| \Smax) = D^{\infty}(\Phi^{(B)} \| S^{(B)})$ for any free state $\mu^{(A)} \in S^{(A)}$.
  Since $A$ can prepare any $\mu^{(A)}$ freely, we have
  \begin{align*}
  &R(\rho^{(AB)} \xrightarrow{(S, F)} \Phi^{(B)})\\
  &=
  R(\rho^{(AB)} \xrightarrow{(S, F)} \mu^{(A)} \otimes \Phi^{(B)})
  \\
  &\leq \frac{D^{\infty}(\rho^{(AB)} \| \Smin)}{D^{\infty}(\mu^{(A)} \otimes \Phi^{(B)} \| \Smax)}
  =
  \frac{D^{\infty}(\rho^{(AB)} \| \Smin)}{D^{\infty}(\Phi^{(B)} \| S^{(B)})},
  \end{align*}
  where we used \cref{theorem:asymptotic} for the inequality.
\end{proof}

We remark that \cref{corollary:assisted-distillation} is completely independent of the composite theory and it holds for general resource theories, in contrast to the known assisted distillation bounds derived for specific resources and composite operations.

A well-known specific case is the assisted distillation of coherence, where system $B$ is restricted to incoherent operations (IO)~\cite{Streltsov_2017a,Baumgratz_2014} and the shared initial state $\psi^{(AB)}$ is pure.
Refs.~\cite{Chitambar_2016a,Streltsov_2017} showed that in the LFOCC framework, the assisted distillation rate is given by 
\begin{equation}
    R(\psi^{(AB)} \xrightarrow{(\Smin, LFOCC)} \Phi^{(B)}) = S(\Delta(\psi^{(B)})),
\end{equation}
where $\Delta(X) = \sum_i \ketbra{i} \ev{X}{i}$ is the dephasing map in the incoherent basis.
Furthermore, observe that $D^{\infty}(\Phi^{(B)} \| S^{(B)}) = D(\Phi^{(B)} \| S^{(B)}) = 1$, since relative entropy of coherence is additive~\cite{Gour_2009,Baumgratz_2014}.
Compare this to \cref{corollary:assisted-distillation}, where we have 
\begin{equation}
    R(\psi^{(AB)} \xrightarrow{(S, F)} \Phi^{(B)})
\leq
D^{\infty}(\psi^{(AB)} \| \Smin).
\end{equation}
Remarkably, Ref.~\cite[Theorem~4]{Chitambar_2016a} showed that $
D(\psi^{(AB)} \| \Smin)
=
S(\Delta(\psi^{(B)}))
=
D^{\infty}(\psi^{(AB)} \| \Smin)
$, demonstrating that our bound is saturated by the LFOCC framework.

In fact, we can make a stronger statement, independent of the resource theory at $B$. Using
\cref{lemma:uncorrelated}, we have
\begin{equation}
    D(\Phi^{(B)} \| S^{(B)}) = D(\mu^{(A)} \otimes \Phi^{(B)} \| \hat S) 
\end{equation}
for $\hat S \in \lbrace \Smin,\Smax \rbrace $ and any $\mu^{(A)}$, which implies
\begin{align*}
	\frac
	{D^{\infty}(\rho^{(AB)} \| \Smin)}
	{D^{\infty}(\Phi^{(B)} \| S^{(B)})}
  &=
    	\frac
	{D^{\infty}(\rho^{(AB)} \| \Smin)}
    {D^{\infty}(\Phi^{(B)} \| \Smin)}
    \\
	&=
          R(\rho^{(AB)} \xrightarrow{(\Smin, ARNG(\Smin))} \Phi^{(B)})
          \\
	&\leq
          \sup_{(S, F)} R(\rho^{(AB)} \xrightarrow{(S, F)} \Phi^{(B)})
          \\
	&\leq
	\frac
	{D^{\infty}(\rho^{(AB)} \| \Smin)}
	{D^{\infty}(\Phi^{(B)} \| S^{(B)})}.
\end{align*}
This shows that, whenever the Brandão–Plenio conditions are satisfied, the bound of \cref{corollary:assisted-distillation} is always achievable by the composite theory $(\Smin,\mathrm{ARNG}(\Smin))$. In this sense, the extremal composite theory not only provides a universal bound, but also attains it.

\subsubsection{Assisted distillation as a witness of correlations}
Interestingly, \cref{corollary:assisted-distillation} also allows us to estimate the correlation present in the system.
Recall the following useful equality for relative entropy: when $\sigma^{(AB)} =\sigma^{(A)} \otimes \sigma^{(B)}$,
\begin{align*}
    D(\rho^{(AB)} \| \sigma^{(AB)}) &= D(\rho^{(AB)} \| \rho^{(A)} \otimes \rho^{(B)}) + D_A+D_B,
\end{align*}
where $D_A = D(\rho^{(A)} \| \sigma^{(A)})$ and similarly for $D_B$.
Because $S^{(A)}$ includes all density matrices, $D(\rho^{(AB)} \| \Smin) \leq D(\rho^{(AB)} \| \rho^{(A)} \otimes \rho^{(B)}) + D(\rho^{(B)} \| S^{(B)})$.
Similarly, $D([\rho^{(AB)}]^{\otimes n} \| \Smin) \leq D([\rho^{(AB)}]^{\otimes n} \| [\rho^{(A)}]^{\otimes n} \otimes [\rho^{(B)}]^{\otimes n}) + D([\rho^{(B)}]^{\otimes n} \| S^{(B)})$, and it follows that the regularized quantity $D^{\infty}(\rho^{(AB)} \| \Smin) \leq D(\rho^{(AB)} \| \rho^{(A)} \otimes \rho^{(B)}) + D^{\infty}(\rho^{(B)} \| S^{(B)})$.
Thence, we obtain a quantitative estimate of the correlations in the system from the assisted distillation rate:
\begin{align}
	&\frac
	{D(\rho^{(AB)} \| \rho^{(A)} \otimes \rho^{(B)})}
	{D^{\infty}(\Phi^{(B)} \| S^{(B)})}
    \nonumber \\
	&\quad \geq
	R(\rho^{(AB)} \xrightarrow{(S, F)} \Phi^{(B)})
	-
	\frac
	{D^{\infty}(\rho^{(B)} \| S^{(B)})}
	{D^{\infty}(\Phi^{(B)} \| S^{(B)})}.
\end{align}
Note that $
{D^{\infty}(\rho^{(B)} \| S^{(B)})}/
{D^{\infty}(\Phi^{(B)} \| S^{(B)})}
$ upper bounds the \emph{unassisted} distillation rate.
In this way,  we have shown that the assisted distillation rate is a witness of correlations.
In fact, it is a faithful one:
whenever the assisted distillation rate is strictly higher than the unassisted rate, the initial state must be correlated.
\begin{lemma}
	Suppose the initial state $\rho^{(AB)} = \rho^{(A)} \otimes \rho^{(B)}$ is uncorrelated.
	Then the assisted distillation rate is equal to the non-assisted rate, i.e.
	\begin{align}
		R(\rho^{(A)} \otimes \rho^{(B)} \xrightarrow{(S, F)} \Phi^{(B)})
		&=
		R(\rho^{(B)} \xrightarrow{(S^{(B)}, F^{(B)})} \Phi^{(B)}).
	\end{align}
\end{lemma}
\begin{proof}
	Clearly, we have
	$
	R(\rho^{(AB)} \xrightarrow{(S, F)} \Phi^{(B)})
	\geq
	R(\rho^{(B)} \xrightarrow{(S^{(B)}, F^{(B)})} \Phi^{(B)})
	$
	since $\mathbf{1}_A \otimes F^{(B)} \subseteq F$.
	Therefore we only need to show
	$
	R(\rho^{(A)} \otimes \rho^{(B)} \xrightarrow{(S, F)} \Phi^{(B)})
	\leq
	R(\rho^{(B)} \xrightarrow{(S^{(B)}, F^{(B)})} \Phi^{(B)})
	$.

    Let us fix an arbitrary $\epsilon, \delta > 0$, and take an achievable rate of assisted distillation $r < R(\rho^{(A)} \otimes \rho^{(B)} \xrightarrow{(S, F)} \Phi^{(B)})$.
    By definition, there exist $m, n$ and $\Lambda^{(AB)}\in F$ such that $\frac{m}{n} \geq r - \delta$, and
    \begin{gather}
	    \norm{\Lambda^{(AB)} \pqty{\bqty{\rho^{(A)}}^{\otimes m} \otimes \bqty{\rho^{(B)}}^{\otimes m}} - \bqty{\Phi^{(B)}}^{\otimes n}}_1
	     \leq \epsilon.
    \end{gather}
    Recalling that the theory at $A$ is trivial and therefore $\rho^{(A)}$ is a free state, condition~\ref{condition:free marginal operations} of \cref{definition: composite RTs} states that the channel $\Lambda^{(B)}(Y) = \Tr_{A}\Lambda^{(AB)}([\rho^{(A)}]^{\otimes n}\otimes Y)$ is a free operation in $F^{(B)}$.
    Since $\epsilon, \delta$ were picked arbitrarily, this shows that $r$ is an achievable rate for the unassisted distillation, i.e. $r \leq R(\rho^{(B)} \xrightarrow{(S^{(B)}, F^{(B)})} \Phi^{(B)})$.
    Taking supremum over all $r < R(\rho^{(A)} \otimes \rho^{(B)} \xrightarrow{(S, F)} \Phi^{(B)})$, we obtain the claim.
\end{proof}

\subsection{Resource certification}

\begin{figure*}
	\includegraphics[width=0.8\textwidth]{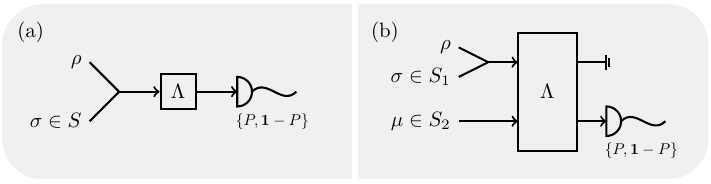}
	\caption{\label{figure:resource-certification}
		Resource certification:
		(a) standard resource certification,
		(b) remote resource certification.
	}
\end{figure*}

We now turn to certifying, or detecting the presence of, quantum resources.
It is instructive to formulate this task within the framework of quantum hypothesis testing~\cite{Gour2024Book} (see \cref{figure:resource-certification}). 
We begin with the standard (non-remote) certification scenario. 
Alice holds a single quantum system governed by a resource theory $(S,F)$. 
She is promised that the system is either in a known resource state $\rho$, or in an unknown free state $\sigma \in S$. 
To decide which is the case, Alice applies a binary POVM $\Bqty{P, \mathbf{1} - P}$: 
outcome $P$ leads her to guess $\rho$, otherwise she guesses that the state is free. 
This decision procedure gives rise to two types of errors:
\begin{itemize}[leftmargin=*,itemsep=0pt]
	\item Type-I error: the state is free ($\sigma \in S$), but Alice incorrectly guesses $\rho$.
	The worst-case probability of this error is $\alpha (P) = \sup_{\sigma \in S} \Tr \pqty{\sigma P}$.
	\item Type-II error: the state is in fact $\rho$, but Alice incorrectly guesses it to be free. 
	This error occurs with probability $\beta (P)=\Tr \pqty{\rho \pqty{\mathbf{1} - P}}$.
\end{itemize}

Various operational setups of this hypothesis-testing task can be considered. 
This includes symmetric testing (minimize average error given priors), asymmetric testing (minimize type-II error subject to an upper bound on type‑I error), and asymptotic tests (study scaling of errors with the number of copies).
For concreteness, we focus on the asymmetric setting, where the figure of merit is
\begin{align}
  D_H^\epsilon (\rho \| S) = -\log \inf_{\substack{0 \leq P \leq \mathbf{1} \\ \alpha(P) \leq \epsilon}} \beta (P),
\end{align}
with a fixed $\epsilon > 0$, which is known as the hypothesis testing relative entropy~\cite{Gour2024Book,Hiai_1991,Ogawa_2005}. 
It quantifies how well $\rho$ can be distinguished from the set of free states and thus provides an operational measure for the resource content of $\rho$.
Indeed, $D_H^\epsilon$ satisfies data-processing inequality and therefore $D_H^\epsilon (\rho \| S)$ is a resource monotone. 

\subsubsection{Remote resource certification}
Let us consider the remote resource certification task. 
Alice holds either the resource state $\rho^{(A)}$ or an unknown free state $\sigma^{(A)} \in S^{(A)}$.
Bob's task is to decide which is the case by performing a local POVM measurement $\Bqty{P^{(B)}, \mathbf{1}_B - P^{(B)}}$ on his subsystem. 
Without additional structure or access to shared correlations, the local POVM does not help Bob at all.

To enable nontrivial strategies, we allow Alice and Bob a joint preprocessing operation $\Lambda^{(AB)}$. 
Consistent with the resource-theoretic constraints, $\Lambda^{(AB)}$ must be a free operation in some composite resource theory $(S, F)$.
In standard (non-remote) hypothesis testing, it follows from data processing inequality that such preprocessing steps cannot enhance the certification task. 
In contrast, preprocessing is necessary for any remote certification strategy that achieves better-than-trivial performance. 
The achievable performance hinges directly on the chosen composition structure. 
Related setups have been considered in Ref.~\cite{Sarkar_2025}, as well as classical settings~\cite{Ahlswede_1986,Ahlswede_2019}.

As in the single-party case, two types of errors arise. Given a measurement $P^{(B)}$ and preprocessing $\Lambda^{(AB)}$:
\begin{itemize}[leftmargin=*]
	\item Type-I error: the system is in some state $\sigma \in S^{(A)}$, but Alice guesses $\rho$.
  The worst-case error is
  \begin{align*}
    \alpha (P^{(B)}, \Lambda^{(AB)})
    &=
    \sup_{\sigma \in S^{(A)}} \Tr \pqty{\Lambda^{(AB)} (\sigma^{(A)}) P^{(B)}}.
  \end{align*}
\item Type-II error: the system is in the state $\rho$, but Alice guesses $S$.
  The error is
  \begin{align*}
    \beta (P^{(B)}, \Lambda^{(AB)})
    &=
    \Tr \pqty{\Lambda^{(AB)}(\rho^{(A)}) \pqty{\mathbf{1} - P^{(B)}}}.
  \end{align*}
\end{itemize}

In the asymmetric setting, the figure of merit is the minimal type-II error with an upper bound on type-I error quantified by
\begin{align}\label{eq:distFeps}
  \dist_F^{\epsilon}(\rho^{(A)} \| S^{(A)})
  &=
  \sup_{\substack{\Lambda \in F^{(AB)}\\ 0 \leq P \leq \mathbf{1}_B \\ \alpha(P, \Lambda) \leq \epsilon}}
  -\log \beta \pqty{P, \Lambda}.
\end{align}
Intuitively, we expect the optimal performance to be bounded by the distinguishability of $\rho^{(A)}$ from $S^{(A)}$.
This is indeed the case:
\begin{lemma}\label{proposition:remote-certification-optimal}
  Let $(S, F)$ be a composite resource theory. Then
  \begin{align}
    \dist_F^{\epsilon}(\rho^{(A)} \| S^{(A)})
  &\leq
    D_H^\epsilon (\rho^{(A)} \| S^{(A)}).
  \end{align}

\end{lemma}
\begin{proof}
  From Eq.~\eqref{eq:distFeps} we can write
  \begin{align*}
      &\dist_F^{\epsilon}(\rho^{(A)} \| S^{(A)})\\
   &=
   \sup_{\Lambda \in F^{(AB)}}
   D^\epsilon_{H} \pqty{ \Tr_A \Lambda(\rho^{(A)}) \| \Tr_A \Lambda(S^{(A)})}.
  \end{align*}
  Using the data processing inequality on $D_H^{\epsilon}$ immediately proves the statement. 
\end{proof}

\subsubsection{A case study for coherence and imaginarity}
To see how the choice of composition affects remote certification, consider an example.
First, let Alice be constrained by the resource theory of strictly incoherent operations (SIO)~\cite{Winter_2016a,Yadin_2016a}. 
Given a preferred basis $\Bqty{\ket{i}}$, the set of SIO free states are those incoherent with respect to this basis, i.e. $\mu = \sum_i p_i \ketbra{i}$.
SIO free operations are those with a Kraus decomposition $\Bqty{K_i}$, such that for any incoherent state $\mu$, the operators $K_i \mu K_i^\dagger$ and $K_i^\dagger \mu K_i$ are all incoherent. 
The first condition implies that under any selective measurement, incoherent states are preserved.
Similarly, by looking at the set of adjoint maps we can interpret the second condition as an analogous requirement on the set of incoherent measurements, i.e. measurements whose POVM elements are diagonal.
Physically, SIO captures the setting in which one cannot prepare coherent states and is restricted to incoherent measurements.

Next, consider Bob being constrained by the resource theory of imaginarity~\cite{Hickey_2018}. 
This is a resource-theoretic formulation of the role of complex numbers in quantum mechanics.
Similar to coherence theory, we start with a preferred basis $\Bqty{\ket{i}}$.
Free states are those whose matrix elements in this basis are real, i.e. $\bra{i}{\mu}\ket{j} \in \mathbb{R}$ for all $i, j$.
Free operations are the real operations, which have several equivalent definitions~\cite{Hickey_2018,Kondra_2023}:
(1) there is a set of Kraus operators $\Bqty{K_i}$ with only real matrix elements,
(2) for any $\rho$, $\Lambda(\rho^T) = \bqty{\Lambda(\rho)}^T$, where ${}^T$ denotes transpose in the preferred basis, and
(3) for any real state $\rho^{(12)}$, $\pqty{\mathbf{1}_1 \otimes \Lambda_2} (\rho^{(12)})$ is a real state.
Finally, we define a set of measurements that can be implemented at no cost, i.e. the set of real measurements.
These are simply measurements whose POVM elements only have real-valued entries.

For simplicity, suppose both Alice and Bob hold a single qubit, and the resource state that is to be certified on Alice's side is $\rho^{(A)} = \ketbra{+y}^{(A)}$, where $\ket{+y}=\frac{1}{\sqrt{2}}(\ket{0}+i\ket{1})$.
Locally, Alice alone cannot distinguish $\rho^{(A)}$ from the set of incoherent states $S^{(A)}$ with incoherent measurements.
Similarly, if she simply sends the state to Bob, he also cannot distinguish the state $\rho^{(A)}$ from the set of incoherent states $S^{(A)}$ with real measurements.
However, if Alice performs a $\pi/2$-rotation along $z$-axis and then sends the state to Bob, then he can perform an $X$ measurement to saturate the bound in \cref{proposition:remote-certification-optimal}.
This preprocessing can be included in the free operation as it always maps Alice's incoherent state into the same incoherent (and thus real) state for Bob.

This example highlights the essential role of joint preprocessing in remote certification: enlarging the set of admissible composite free operations can strictly improve performance. In \ref{supp:remote-certification}, we analyze this example in detail and demonstrate a quantitative gap between LFOCC-based preprocessing and RNG-based strategies.

\section{Discussions and Conclusion}

We have developed an axiomatic theory of heterogeneous quantum resources on networks. 
The central strength of this approach lies in its flexibility: we impose only minimal consistency requirements that exclude the pathological generation of resources for free, without committing to specific capabilities for quantum communication or global operations across the network. 
In fact, we provide results for both classical networks (LFOCC) and fully quantum ones constrained only by the scarcity of resources (RNG), and any realistic scenario would lie between these extremes.  

In this general setting, resource manipulation acquires qualitatively new features in heterogeneous networks. 
Distinct resources are governed by different mathematical structures, and that composing or hybridizing such structures introduces genuinely new constraints and possibilities for state transformations.
We demonstrate this explicitly through operational tasks, such as resource conversion, assisted distillation, and remote certification. 
Instead of adopting implicit or theory-specific composition rules, as is common in prior works, we identify the extra degrees of freedom inherent in the act of composing quantum resources.

This perspective yields new insights even for the special case of homogeneous resources on networks. 
For instance, Ref.~\cite{Son_RobustCat} classified resources by extensiveness, the possibility of catalytic operations, and broadcastability assuming conventional composition rules for each resource. 
However, our network-based formulation shows that these operational properties depend not only on the intrinsic nature of a resource, but also on the compositional structure imposed by the network. In this sense, composition is elevated from a \emph{background, context-dependent choice} into a \emph{central, new physical ingredient} that directly influences the laws of multipartite resource manipulation.

Building on this viewpoint, we derived universal bounds for simultaneous manipulation of multiple resources in both single-shot and asymptotic regimes. Crucially, our bounds are independent of any particular composite theory and depend only on local resource structures. Hence, they establish general laws that apply across a wide range of network architectures. We view these results as a first step toward an emerging theory of multi-resource interconversion, where a wealth of trade-offs, synergies, and incompatibilities between resources can be systematically quantified.

We highlight a few specific directions that naturally emerge from this work. 
One immediate opportunity is to develop a resource-theoretic model of quantum thermodynamic systems having access to multiple heat baths at different temperatures.
Typical resource‑theoretic treatments struggle to accommodate multiple infinite heat baths within a single framework.
While this scenario has been analyzed with other approaches~\cite{Scovil_1959,Mazza_2014,Kosloff_2014,Vinjanampathy_2016}, a resource-theoretic model would allow us to derive relevant laws of thermodynamics. 

Another interesting direction concerns catalysis---the phenomenon where the presence of a resource state enables an otherwise forbidden transformation~\cite{LipkaBartosik2023CatReview}. 
Traditionally, catalytic transformations are studied under the assumption that both target system and catalyst carry the same type of resource.
Our framework removes this restriction and opens the door for studying \emph{heterogeneous} catalysis, where the catalyst carries a different resource from the target system.
Although not explicitly formulated, some examples of robust catalysis identified in Ref.~\cite{Son_RobustCat} hint at this possibility.
However, a systematic investigation of how heterogeneous catalysts enhance performance or enable otherwise forbidden transformations is an intriguing, largely unexplored direction.

Finally, the network-based perspective developed here naturally accommodates genuinely multipartite tasks, such as resource allocation~\cite{Salazar2021optimalallocationof} and multipartite resource broadcasting. More broadly, our framework supports complex architectures in which a hypergraph of distinct resource theories governs a quantum network, which is an abstraction well suited for near-term quantum internets. In such settings, the tools introduced in this work, including remote resource certification and monotones translated across resource theories, are expected to play a central role in optimizing distributed quantum processing capabilities.

\section{Acknowledgements}
The authors thank Chung-Yun Hsieh, Hyukjoon Kwon Soojoon Lee, Paul Skrzypczyk, and Ryuji Takagi for helpful comments and discussions.
RG, JS, JC and NN are supported by the start-up grant of the Nanyang Assistant Professorship at the Nanyang Technological University in Singapore. NN acknowledges the National
Research Foundation, Singapore (W24Q3D0001) and the Tier 1 MOE grant RT1/23 ``Catalyzing quantum security: bridging between theory and practice in quantum communication protocols".
S.H.L is supported by Global Partnership Program of Leading Universities in Quantum Science and Technology (RS-2025-02317602) and Institute of Information and Communications Technology Planning and Evaluation (IITP) Grants (No. RS-2025-25464492, "Hybrid quantum resource theories for robust quantum catalysts", No. RS-2025-02283189, ``Quantum Metrology Theory Based on Temporal Correlation Operators''), and the start-up grant of Ulsan National Institute of Science and Technology, South Korea.
\bibliography{composite_ref.bib}

\clearpage
\newpage
\setcounter{page}{1}
\title{Supplemental Materials for ``Manipulating heterogeneous quantum resource over a network"}
\renewcommand{\thetheorem}{S\arabic{theorem}}
\setcounter{theorem}{0}
\renewcommand{\thelemma}{S\arabic{lemma}}
\setcounter{lemma}{0}
\renewcommand{\thecorollary}{S\arabic{corollary}}
\setcounter{corollary}{0}
\renewcommand{\theequation}{S\arabic{equation}}
\setcounter{equation}{0} 
\renewcommand{\theremark}{S\arabic{remark}}
\setcounter{remark}{0} 
\renewcommand{\theproposition}{S\arabic{proposition}}
\setcounter{proposition}{0} 
\renewcommand{\thedefinition}{S\arabic{definition}}
\setcounter{definition}{0} 
\setcounter{section}{0}
\renewcommand{\thesection}{Supplementary Note \arabic{section}}
\maketitle
\onecolumngrid

\section{Overview of existing approaches to multi-resource theories}\label{supp:related-works}

The interaction between different resource theories has been previously studied in several different forms. We discuss them in conjunction with our framework here:

\begin{enumerate}[leftmargin=*]
	\item \textbf{Multiple resource theories imposed \emph{on the same system}.} --- Closest in spirit to our work is the axiomatic multi-resource theory framework of Refs.~\cite{Sparaciari_2020,Sparaciari_2017}. Concretely, suppose that $\pqty{S_1, F_1}, \pqty{S_2, F_2}$ are two different resource theoretic structures, capturing different operational limitations.
	In this framework, one is subject to both limitations, so the only allowed operations are \mbox{$F = F_1 \cap F_2$}, such that they preserve both $S_1$ and $S_2$. A prototypical example of this framework is the GPC model of quantum thermodynamics, where model thermodynamically free operations are modelled by both Gibbs-preservation (GP) and covariance with respect to time-translation~\cite{Cwiklinsky2015EnTO}. Within this framework, it was shown that it is possible to convert between different resources, through a construction of \emph{batteries} and \emph{bank states}.
	Batteries play the mediating role of storing and releasing resources of a certain type, while bank states enable the interconversion of resources in different batteries, i.e. $\textrm{bank}^{\otimes n} \otimes \textrm{battery}_1 \otimes \textrm{battery}_2 \to \textrm{bank}'^{\otimes n} \otimes \textrm{battery}'_1 \otimes \textrm{battery}'_2$ (see the discussion around Eq.~(40) in Ref.~\cite{Sparaciari_2020}).
	Furthermore, they derive a general upper bound on this resource conversion, which is dubbed the \emph{first law of general quantum resource theories}. There is still much potential in discovering intricate structures in such settings, as recently explored in \cite{koukoulekidis2025symmetry,hu2026gaussian} for non-Gaussianity and asymmetry.
	
	The main difference between the framework of Ref.~\cite{Sparaciari_2020,Sparaciari_2017} and this work lies in the choice of simplifying assumptions.
	The former focuses on specific resource theories possessing convenient mathematical properties, in particular those that admit reversible manipulation. For example, this is implied by the requirement that the theories satisfy the ``asymptotically equivalent property'' (see Definition 1).
	While many resource theories admit such reversible variants for theoretical analysis, some physically important examples—most notably entanglement—do not~\cite{Lami2023}, unless under specific assumptions~\cite{ganardi2025second}. This limits the applicability of that framework to particularly well-behaved resources.
	In contrast, we adopt a distributed setting with multiple parties, where each local subsystem is equipped with its own resource-theoretic structure. We show that this additional structure allows us to derive novel results such as resource conversion laws, without \textit{a priori} assuming reversibility.
	Furthermore, we show that our distributed setting gives rise to improved \emph{tight} bounds on tasks such as assisted distillation and novel scenarios such as remote resource certification.
	
	\item \textbf{Resource engines (sequential composition of free operations)}. --- A contrasting approach is taken in Ref.~\cite{WS2024}.
	Here, the different resource theories are \emph{fused} in an engine-like fashion, where a free operation can consist of many strokes, and in the $i$-th stroke we can perform a free operation from the $i$-th theory.
	The primary motivation here is to model a heat engine, where the system can be coupled to a bath at a different temperature during each stroke.
	The paper focused on studying this model in several interesting scenarios, such as thermodynamics with two different baths and coherence theory with two different bases. In a similar spirit, one can consider the capability of one set of free operations in generating a second type of resource~\cite{de2024entanglement,deneris2025analyzing}.

		This setting is fundamentally different from ours, as each stroke acts on a single system. Moreover, resource creation is generally allowed, since a free state of one theory need not be free in another. While physically well motivated and promising, this approach explores a complementary aspect of combining resource theories. From the perspective of our framework, it can be viewed as corresponding to the strong assumption that unlimited free quantum communication is available between different parties in a network.
	
	\item \textbf{Extended free state sets (allowing correlations, or union of free states)} --- A complementary approach comes from the consideration of thermodynamical systems that might be correlated~\cite{Bera_2017a,Bera_2021}.
	It has been shown that these correlated systems might break the usual laws of thermodynamics, manifesting new phenomena such as anomalous heat flow~\cite{Jennings_2010} and violation of the Landauer bound~\cite{Rio_2011}.
	Ref.~\cite{Bera_2017a} considers systems that are correlated with the bath, and derives general laws of thermodynamics under entropy-preserving operations that remain valid in this regime.
	Ref.~\cite{Faist_2019,Sagawa_2021} consider transformations between multipartite, interacting systems whose thermal state is correlated.
	Ref.~\cite{Bera_2021} extends the usual resource theory framework for thermodynamics to handle multiple baths at different temperature, and derive a generalized second law statement.
	
	\item \textbf{Local restrictions on quantum operations in LOCC frameworks. ---} There is also extensive work in entanglement theory with additional resource-theoretic restrictions on the local operations, such as Gaussian operations~\cite{Fiurasek_2002,Eisert_2002,Giedke_2002}, incoherent operations~\cite{Chitambar_2016b,Streltsov_2017}, thermal operations~\cite{arxiv_Bistron_2024,Horodecki_2003b}, unitals~\cite{Ganardi_2025a}, or stabilizer operations~\cite{Andi_Gu_2025}.
	Our results encompass these settings, but also extend beyond non-LOCC based settings, such as the framework considered in~\cite{Bera_2021}.
	A simple example of the setting that is covered by our framework but not LOCC-based results is the following: suppose that in addition to locally free operations and classical communication, parties also have access to locally-free entanglement.
	This is beyond the scope of LOCC-based results, but is well-motivated in some settings (a network of photonic quantum computers might share two-mode squeezed vacuum states freely) and can be treated appropriately in our framework.
\end{enumerate}

\section{Technical preliminaries and definitions}\label{supp:preliminaries}

We consider quantum systems living in finite $n$-dimensional Hilbert spaces denoted by $\mathcal{H}\cong \mathbb{C}_n$.
Quantum states are represented by positive semidefinite operators with unit trace, which we denote by $\mathcal{S}(\mathcal{H})$, while quantum channels are completely positive, trace-preserving (CPTP) maps, which is denoted by $\mathcal{O}(\mathcal{H} \to \mathcal{H'})$.
For identical input and output systems, we write $\mathcal{O}(\mathcal{H})$ for brevity.

Most quantum resource theories are defined not only on a single system $\mathcal{H}$, but also on any composite system $\mathcal{H} \otimes \mathcal{H'}$.
As such, technically we have to define the set of free states $S$ as a \emph{function} of the Hilbert space, where for each Hilbert space $\mathcal{H}$, there is a subset of states $S(\mathcal{H}) \subseteq \mathcal{S}(\mathcal{H})$ that is free.
Similarly, the set of free operations is a function of input and output Hilbert spaces, i.e. $F(\mathcal{H} \to \mathcal{H'}) \subseteq \mathcal{O}(\mathcal{H} \to \mathcal{H'})$. We will often drop the underlying Hilbert spaces from $S(\mathcal{H}), F(\mathcal{H} \to \mathcal{H'})$ when it is clear from the context.

When considering resource theories on a composite system of multiple parties $\mathcal{H} = \bigotimes_i \mathcal{H}^{(i)}$, we denote the \emph{local} resource theories acting on $\mathcal{H}^{(i)}$ by an upper index $(S^{(i)}, F^{(i)})$, that is, for each $i$ we have a set of free states and operations:
\begin{equation}
	S^{(i)} \subseteq \mathcal{S}(\mathcal{H}^{(i)}) ,\qquad F^{(i)} \subseteq \mathcal{O}(\mathcal{H}^{(i)}).
	\end{equation}
We reserve the notation $(S, F)$ for an arbitrary \emph{composite} resource theory acting on the whole system $\mathcal{H}$ (see \cref{definition: composite RTs}).

Throughout the manuscript, we draw on several resource theories that have been extensively studied in the literature for illustration and example. For detailed overviews of these theories, we refer the reader to the corresponding reviews~\cite{Bartlett2007, Horodecki_2009a, Streltsov_2017a, Lostaglio_2019, Chitambar_2019}. For a given resource theory $(S, F)$, we say that $\rho$ can be transformed to $\sigma$ if there exists $\Lambda \in F$ such that $\Lambda(\rho) = \sigma$, or more compactly $\rho \xrightarrow{(S, F)} \sigma$.
We say that $R: \mathcal{S}(\mathcal{H}) \to \mathbb{R}$ is an $(S, F)$-monotone if for any $\rho, \sigma \in \mathcal{S}(\mathcal{H})$ such that $\rho \xrightarrow{(S, F)} \sigma$, we have $R(\rho) \geq R(\sigma)$.
For any monotone, the minimal value over all states is always reached on the set of free states (if the set of free operations include the preparation of any free states), so by convention we usually set $R(\sigma) = 0$ for any $\sigma \in S$.
When $R(\sigma) = 0$ if \emph{and only if} $\sigma \in S$, then we say that $R$ is a \emph{faithful} monotone.

In addition to the aforementioned single-shot setting, typically, we are also interested in transformations of multiple copies of a certain state.
Thus, if the state of a single copy of the system lives in $\mathcal{S}(\mathcal{H})$, we implicitly assume that for any $n$, there is a naturally defined set of free states $S_n \subseteq \mathcal{S}(\mathcal{H}^{\otimes n})$, and analogously for $F$
(note that the index for multiple copies is at the bottom).
With a slight abuse of notation, we will denote $\bigcup_{n=1}^\infty {S_n}, \bigcup_{n=1}^\infty{F_n}$ as $S, F$ when it is clear from the context that we are considering multi-copy transformations.

In many theories, any two arbitrary states $\rho, \sigma$ can be transformed in this regime, i.e. we have $\rho^{\otimes m} \xrightarrow{(S, F)} \sigma^{\otimes n}$ for some appropriate $m, n \geq 1$.
Therefore, in this setting, the main quantity of interest is the \emph{rate} of transformation.
We say that $r$ is an achievable rate for the asymptotic transformation $\rho$ to $\sigma$ if for any $\epsilon, \delta > 0$, there exist $m, n \in \mathbb{N}$ and $\Lambda \in F$ such that
\begin{align}
  \norm{
  \Lambda(\rho^{\otimes m}) - \sigma^{\otimes n}
  }_1 &\leq \epsilon,
  \qquad 
  \frac{m}{n} \geq r - \delta.
\end{align}
We denote the supremum of all achievable $r$ by $R(\rho \xrightarrow{(S, F)} \sigma)$.

\section{Composing resource theories}\label{supp:composition}
Section \ref{sec:composing} lies out our axiomatic approach towards composing resource theories in full generality. In this section, we elaborate in detail on the simplest composite theory, i.e. entanglement theory, which is extremely well-studied and serves as a guiding case study. We then briefly review and discuss composition of homogeneous resources which are usually implicitly assumed for distinct quantum resources, in which the Brandão-Plenio axioms serve as a key guiding set of principles. The last subsection contains a few extra technical remarks for the composition of heterogeneous resources.

\subsection{Baby's first composite resource theory: Entanglement}

Entanglement theory is perhaps the prototype of a theory of composite systems, with a focus on the non-classical correlations that such systems can have~\cite{Horodecki_2009a}.
In standard entanglement theory, we assume that we can implement any local operation (LO) on any party.
Additionally, we typically allow classical communication (CC) which allows a party to implement a local operation condition on the outcome of a measurement performed by a different party.
The set of (free) operations that can be implemented in such a way is called local operations and classical communications (LOCC)~\cite{Bennett_1996b,Donald_2002,Chitambar_2014}.
The set of free states are the so-called separable states~\cite{Werner_1989}, i.e. the convex combination of product states $\mu = \sum_j p_j \bigotimes_{i} \rho_{j}^{(i)}$.
Together, they give rise to the laws of entanglement manipulation, which have been highly influential in the field of quantum information.
Some entanglement quantifiers such as entanglement entropy and purity of entanglement have even found applications in fields like many-body physics~\cite{Eisert_2010}, black hole physics~\cite{Srednicki_1993,Bombelli_1986}, and high-energy physics~\cite{Ryu_2006}.
The techniques developed in entanglement theory have shaped and will continue to shape our understanding of quantum resources and how to quantify them.

Due to the difficulty of handling LOCC, other variants that modify the set of free operations have been considered in the literature.
We can weaken the set of free operations to only allow one-way classical communication (1-LOCC)~\cite{Bennett_1996b}, or even completely replace it with a source of shared randomness (LOSR)~\cite{arxiv_Dukaric_2008,Buscemi_2012a,Forster_2009,Schmid2020typeindependent,Schmid_2023}.
In the other direction, we can consider the set of separable operations or non-entangling operations~\cite{Harrow_2003}.
Although the set of non-entangling operations does not have a clear physical interpretation, it acts as a universal upper bound to any physically-motivated free operations.
This has been used to show that entanglement theory is fundamentally irreversible~\cite{Lami2023}.
In the asymptotic setting, this role is played by a different set of operations, i.e. the asymptotically non-entangling operations~\cite{Brandao_2010a}.

We can extend this standard entanglement theory to take into account some resource-theoretic constraints on the local systems.
One way is to require that each local operation that is performed by any party must be a free operation according to the local resource theory.
Note that the local theories do not have to be identical, i.e. the systems can be heterogeneous.
With this, we may construct the framework of locally-free operations and classical communication (LFOCC).
Various instances of this framework have been studied in the literature, including restrictions on purity~\cite{Horodecki_2003b,Synak-Radtke_2005,arxiv_Horodecki_2005a}, Gaussianity~\cite{Fiurasek_2002,Eisert_2002,Giedke_2002}, coherence~\cite{Chitambar_2016b,Streltsov_2017}, thermodynamics~\cite{arxiv_Bistron_2024,Ganardi_2025a}, and magic~\cite{Andi_Gu_2025}.

\begin{definition}[Locally-free operations and classical communication (LFOCC)]
    A channel $\Lambda \in \mathcal{O}(\bigotimes_i \mathcal{H}^{(i)})$ is implementable by an $N$-round LFOCC protocol if there exist a sequence of parties $\Bqty{i_n}_n$ and a sequence of Kraus operators $\Bqty{K_{l_n | l_{n-1} \ldots l_0}}_n$ such that:
    \begin{enumerate}[itemsep=0pt]
    \item for all $n < N$, $K_{l_n | l_{n-1} \ldots l_0}$ acts nontrivially only on party $i_n$,
    \item $\Lambda (X) = \sum_{l_N, \ldots, l_0} \pqty{K_{l_N | l_{N-1} \ldots l_0} K_{l_{N-1} | l_{N-2} \ldots l_0} \ldots K_{l_0}} X \pqty{K_{l_N | l_{N-1} \ldots l_0} K_{l_{N-1} | l_{N-2} \ldots l_0} \ldots K_{l_0}}^\dagger$.
    \item for all $n < N$ and for each $l_{n-1}, \ldots, l_0$, the channel defined by the Kraus operators $\Bqty{K_{l_n | l_{n-1} \ldots l_0}}$ is a free operation on party $i_n$,
    \end{enumerate}
    A channel is implementable by an LFOCC protocol if it is implementable by $N$-round LFOCC protocol for some $N$.
    For brevity, we will call such channels LFOCC.
\end{definition}

As an example, suppose we have two parties $A, B$.
Then we say that $\Lambda$ is implementable by $1$-round LFOCC if there exists a set of Kraus operators $\Bqty{A_i}$ that defines a free operation on $A$ and another set of Kraus operators $\Bqty{B_{j|i}}$  that defines a free operation on $B$ \emph{for each $i$} such that $\Lambda(\rho^{(12)}) = \sum_{ij} \pqty{A_i \otimes B_{j|i}} \rho^{(12)} \pqty{A_i \otimes B_{j|i}}^\dagger$.
Intuitively, this means that $A$ performs a local operation that outputs an outcome $i$, which is sent to $B$.
Depending on this message from $A$, $B$ applies the local operation $Y \mapsto \sum_j B_{j|i} Y B_{j|i}^\dagger$.
By allowing more rounds of communications, we obtain the set of $N$-round LFOCC.
Note that by dropping condition 3, we recover the LOCC framework.

\subsection{Composite systems of homogeneous resources}\label{section:homogeneous-composition}

Without calling it explicitly, the standard resource theory formalism has dealt with composite systems with homogeneous resources.
By homogeneous, we mean that each system is governed by the same type of resource, e.g. coherence, non-Gaussianity, magic, etc.
This is done every time we consider multi-copy transformations, catalysis, or whenever we consider the resources in several systems.
To see the connection to our framework, for example, we can simply associate each subsystem with a different node on the network.
However, this is usually done implicitly as there is a natural choice of a composite theory.
For instance, suppose we have two systems $\mathcal{H}^{(1)}, \mathcal{H}^{(2)}$ that are governed by Gibbs-preserving theory, with free states $\gamma^{(1)}, \gamma^{(2)}$, respectively.
Then, the composite system $\mathcal{H}^{(12)}$ would be governed by another Gibbs-preserving theory, with a \emph{composite} free state $\gamma^{(12)} = \gamma^{(1)} \otimes \gamma^{(2)}$.
Similarly, if $A^{(1)} B^{(1)}$ and $A^{(2)} B^{(2)}$ are governed by entanglement theory (in the $A|B$ cut), then the composite system $A^{(1)} A^{(2)} B^{(1)} B^{(2)}$ has a natural definition of composite free states (i.e. separable states in the $A^{(1)} A^{(2)} | B^{(1)} B^{(2)}$ cut) and free operations (i.e. LOCC in the $A^{(1)} A^{(2)} | B^{(1)} B^{(2)}$ cut).
Note that the structure of the composite theory might differ between different resource theories;
for Gibbs-preserving theory, the composite free state is just a tensor product of locally free states, while in entanglement theory, the composite separable states includes correlated (in the $A^{(1)} B^{(1)} | A^{(2)} B^{(2)}$ cut) states such as $\ket{Bell_{A^{(1)} A^{(2)}}} \otimes \ket{Bell_{B^{(1)} B^{(2)}}}$.

It has only been recently noticed that these natural compositions give rise to interesting consequences to particular resource theories.
One of the most prominent consequences is the possibility of resource broadcasting~\cite{Hickey_2018}---creating a resourceful state using an auxiliary state, without any net change to the auxiliary---which is otherwise denied for most minimally composed resource theories~\cite{Son_RobustCat}.
This suggests that the composite structure plays a crucial role in determining the extensibility of a resource.

When we consider composite systems in any resource theory, we often implicitly require that the set of free states $\Bqty{S_n}$ satisfies the Brandão-Plenio axioms~\cite{Brandao_2010,arxiv_Hayashi_2024a,Lami_2025}:

\begin{definition}[Brandão-Plenio axioms]\label{def:BPaxioms}	
	A resource theory $(S,F)$ satisfies the Brandão-Plenio axioms if:
	\begin{enumerate}[itemsep=0pt]
		\item Each $S_n$ is convex and closed.
		\item Each $S_n$ contains a full-rank state.
		\item If $\rho \in S_{n+1}$, then $\Tr_i \rho \in S_n$ for every $i = 1, \ldots n+1$.
		\item If $\rho \in S_n$ and $\sigma \in S_m$, then $\rho \otimes \sigma \in S_{n+m}$.
		\item If $\rho \in S_n$, then $P_{\pi} \rho P_{\pi} \in S_n$ for every permutation $\pi$.
	\end{enumerate}
\end{definition}

These axioms ensure that the set of free states $\Bqty{S_n}$ is stable under tensor product operations;
resources that do not satisfy these axioms often behave in unintuitive ways.
For example, the resource theory of Bell non-locality~\cite{Palazuelos_2012} and genuine multipartite entanglement~\cite{Yamasaki_2022,Palazuelos_2022} are known to display activation:
there exists a state $\rho$ that is free, but $\rho^{\otimes 2}$ is not free.
To some degree, this shows an incompatibility between the theory for composite systems and single systems.
Nevertheless, most resource theories in the literature (including entanglement, coherence, magic, thermodynamics, asymmetry, etc.) \emph{do} satisfy these conditions.
As a side remark, we note that the Brandão-Plenio axiom (4) ensures that the limit in the definition of the regularized relative entropy of a resource exists.
To summarize, the Brandão-Plenio axioms impose a condition on the composite theory to be compatible with the single system theory.

In addition to ensuring stability, another reason for requiring the Brandão-Plenio axioms is that under these assumptions, we can show that there is a \emph{tight} bound on the transformation rate~\cite{Brandao_2008,Brandao_2010a,Lami_2025,arxiv_Hayashi_2024a} (see~\cite{Donald_2002,Horodecki_2002} also for an earlier argument establishing a not necessarily tight bound in entanglement theory):
\begin{align}
	R(\rho \xrightarrow{(S, F)} \sigma)
	&\leq
	\frac
	{D^{\infty} (\rho \| S)}
	{D^{\infty} (\sigma \| S)},
\end{align}
where tightness means for any set of free states $S$, there is a construction of a set of free operations $F$ that achieves the upper bound with equality.
These are known as asymptotically resource non-generating operations (ARNG)~\cite{Brandao_2010a}, i.e. the \emph{maximal} set of transformations that preserves the set of free states asymptotically.

In some resource theories, multiple compatible composite theories might exist.
An illustrative example can be found in the theory of thermodynamics.
Suppose our single system consists of a qubit with some fixed temperature and Hamiltonian.
Then, the set of free states consists of a single state, namely the Gibbs state.
The standard definition of free operations is the set of thermal operations~\cite{Janzing_2000}, which is defined in a bottom-up fashion.
We can also define it axiomatically~\cite{Cwiklinsky2015EnTO}, which is a completely equivalent formulation for a single qubit. 
However, the situation changes for composite (i.e. higher dimensional) systems, where it was shown that these two formulations give rise to different non-equivalent resource theories~\cite{Ding2021GPCvsTO}.
Incidentally, a different approach to composite theory has been considered recently where only the resources contained in a \emph{local} subsystem can be utilized~\cite{arxiv_Zhang_2025c,Horova_2022}, as opposed to previous works that extract resources encoded in the correlations between different subsystems~\cite{Francica_2017,arxiv_Biswas_2025,Perarnau-Llobet_2015,Dahlsten_2011,Manzano_2018,Funo_2013,Maruyama_2005,Levitin_2011}.

It is worth noting that these naturally defined composite theories are often \emph{not} compatible with LOCC.
For example, in the resource theory of magic, Clifford unitaries include manifestly entangling operations such as CNOT.
However, LOCC includes \emph{all} local unitaries, including non-Clifford ones.
This is simply because they capture different resources; LOCC captures entanglement while Cliffords capture magic.
This can be partly rectified by restricting ourselves to LFOCC, but this precludes studying networks where some restricted form of quantum communication is allowed.
For example, in a network of quantum computers, it is reasonable to allow different nodes to share some entanglement.
This is not unique to the resource theory of magic;
some forms of entanglement such as two-mode squeezed states can be created with Gaussian operations, which corresponds to free operations in the resource theory of non-Gaussianity.

\subsection{Additional remarks for heterogeneous resources}\label{supp:heterogeneous}

\begin{remark}
We can state condition (d) in \cref{definition: composite RTs} in an alternate form: 
\begin{quote}
    (d') if $\Lambda^{(i)} \in F$ acts only on $\mathcal{S}(\mathcal{H}^{(i)})$, then $\Lambda^{(i)} \in F^{(i)}$.
\end{quote}
Clearly, condition (d') is implied by condition (d).
Now, suppose that (1) $F$ is closed under concatenation, (2) preparation of any free state is a free operation in all the local theories, and (3) partial trace is a free operation in all the local theories.
We can show that under these assumptions,
(d') now implies (d). Note that the marginal channel $\Lambda^{(i)} (X) = \Tr_{\overline{i}} \Lambda(X^{(i)} \otimes \bigotimes_{j \neq i} \rho^{(j)})$ can be written as a concatenation of three channels: preparation of a free state $\rho^{(j)}$, $\Lambda$, and partial trace.
By assumption, this channel is in $F$.
Therefore by condition (d'), $\Lambda^{(i)} \in F^{(i)}$, and condition (d) is satisfied.
This clarifies the physical content of condition (d): locally, any composite free operations reduces to some local free operation.
\end{remark}

\begin{remark}
While it is clear how to take marginals of a quantum state, there are several inequivalent ways to do so for quantum channels.
For example, we can define it through the procedure implied in condition (d):
to obtain the effective local channel at party $1$, we input a free state on all other parties and trace out the output.
This defines a quantum channel at party $1$, and condition (d) requires that this channel be free.
On the other hand, we could also consider the Choi state of the channel and take the appropriate marginals in that space.
This corresponds to feeding the maximally mixed state $\mathbf{1}_i/d_i$ on all other parties and tracing out the output.
Operationally, the marginal channel obtained through this procedure is the \emph{average} channel that results from inputting a random pure state on all other sites and tracing the output, since $\int d\mu(\psi) \ketbra{\psi}^{(i)} = \mathbf{1}_i/d_i$.
In this manuscript, we choose to adopt the former, as it is a weaker assumption satisfied by virtually all operationally-defined resource theories.
\end{remark}

\section{A uniquely maximal composition does not exist}\label{supp:no-fmax}
\begin{example}[No maximal set of composite free operations]\label{example:no-fmax}
	Suppose we have two parties, and the local resource theories are instances of the theory of unital maps.
	Then $S^{(1)} = \Bqty{\mathbf{1}/d_1}, S^{(2)} = \Bqty{\mathbf{1}/d_2}$, and $F^{(1)}, F^{(2)}$ are the corresponding sets of unital maps.
	For simplicity, let us assume $d_1 = d_2 = d$.
	The extremal composite free states are $\Smin = \Bqty{\mathbf{1}/{d^2}}$ and $\Smax = \Bqty{\rho^{(12)} \,|\, \rho^{(1)} = \rho^{(2)} = \mathbf{1}/d }$.
	Consider the following two resource theories: $(\Smin, \RNG(\Smin))$ and $(\Smax, \RNG(\Smax))$.
	We can easily verify that they are both valid composite theories.
	Thus, if there exists a maximal set of composite operations $\Fmax$, we must have 
    \begin{equation}
	    \RNG(\Smin)  \subseteq \Fmax\ \textrm{and}\ \RNG(\Smax) \subseteq \Fmax.
    \end{equation}
    However, given the maximally entangled state $\ket{\Phi} = \sum_{i < d} \frac{1}{\sqrt{d}} \ket{ii}^{(12)}$, let us consider the map
    \begin{align}
    \Lambda(\rho^{(12)}) = \pqty{\Tr_{12} \rho^{(12)}} \ketbra{\Phi},
    \end{align}
    
    Observe that the state $\Phi$ is in $\Smax$, so the channel $\Lambda$ is in $\RNG(\Smax)$.
	Next, let $U$ be the unitary channel that rotates $\ketbra{\Phi}$ to $\ketbra{00}$, which is in $\RNG(\Smin)$ since it is a unitary and thus is also a unital operation. By the condition in Definition \ref{def:RT} that a resource theory should be closed under concatenation, the map $\Lambda'(\rho^{(12)}) = U \Lambda(\rho^{(12)}) U^{\dagger} = \pqty{\Tr_{12} \rho^{(12)}} \ketbra{00}$ should be in $\Fmax$.
	However, $\Lambda'$ cannot be a free composite operation according to \cref{definition: composite RTs}, since it violates condition (d), i.e. $\Tr_2 \Lambda'(X \otimes \mathbf{1}^{(2)}/d) = \pqty{\Tr X} \ketbra{0}$ is not a unital map.
\end{example}

The example above shows the incompatibility between preparing entangled but locally free states and free operations that preserve only the minimal set of free states. 
One might ask whether, if we restrict to operationally motivated composite free operations, we \emph{can} incorporate locally free entanglement. 
However, a slight modification of the above example shows that adding locally free entanglement to LFOCC also renders the theory trivial.
Under the same setting as \cref{example:no-fmax}, quantum teleportation is implementable with LFOCC because the teleportation protocol can be implemented by projective measurements and conditional unitaries, both of which are unital maps.
Thus, each step below is a free operation:
1) trace out the initial state,
2) prepare two copies of the Bell state $\Phi^{(12)} \otimes \Phi^{(1'2')}$, and
3) teleport subsystem $2$ of $\Phi^{(12)}$ to party $1$ using LFOCC and $\Phi^{(1'2')}$.
This protocol implements a state preparation channel that prepares a pure state on party $1$, and therefore it \emph{cannot} be free.
The example demonstrates that, in the composite setting, multiple incompatible composite resource theories can arise. 
If we instead allow only locally free operations (without classical communication) in addition to locally free entanglement, a proper composite resource theory can be defined.
This is reminiscent of the theory of Bell non-locality, where distinct parties may share a common source of correlations that are then processed only locally~\cite{Buscemi_2012a}.

\section{Constructing monotones from a resource conversion task}\label{supp:monotone}

We begin by proving \cref{proposition:monotone}.
\begin{proof}[Proof of \cref{proposition:monotone}]
  Let $\Lambda^{(1)} \in F^{(1)}$.
  Then,
  \begin{align}
    R_1(\Lambda^{(1)}(\rho^{(1)}))
    &=
      \sup R_2\pqty{\Tr_1 \bqty{\Lambda(\Lambda^{(1)}(\rho^{(1)}) \otimes \mu^{(2)})}}\nonumber
    \leq
      \sup R_2\pqty{\Tr_1 \bqty{\Lambda'(\rho^{(1)} \otimes \mu^{(2)})}}\nonumber
    = R_1(\rho^{(1)}),
  \end{align}
  where we use the fact that for any composite theory $(S, F)$ and $\Lambda \in F$, we have $\Lambda' = \Lambda \circ \pqty{\Lambda^{(1)} \otimes \mathbf{1}_2}\in F$ because $F^{(1)} \otimes \mathbf{1}_2 \subseteq F$ and $F$ must be closed under concatenation.
\end{proof}

Let us examine an exemplary application of \cref{proposition:monotone}.
Consider an $R_2$ monotone defined as the generalized robustness, $R_2(\rho) = D_{\rm max} (\rho \| S^{(2)})$, where $D_{\rm max} (\rho \| \sigma) = \inf\Bqty{ \log \lambda \,|\, \rho \leq \lambda \sigma}$.
Notice that for any $\mu^{(2)} \in S^{(2)}$, any composite theory $(S, F)$ and any $\Lambda \in F$, we have
\begin{align}
  D_{\rm max}(\rho \| S^{(1)})\geq
    D_{\rm max}(\rho \otimes \mu^{(2)} \| \Smin)\nonumber
  \geq
    D_{\rm max}(\Lambda(\rho \otimes \mu^{(2)}) \| \Smax)\nonumber
  \geq
    D_{\rm max}(\Tr_1 \Lambda(\rho \otimes \mu^{(2)}) \| S^{(2)}) = R_1(\rho),
\end{align}
where we use \cref{theorem:single-shot} for the second inequality, and for the last equality we recall $R_1$ is defined by Eq.~\eqref{eq:2monotone}.
This shows that $R_1(\rho) \leq D_{\rm max} (\rho \| S^{(1)})$.
Thus we obtain a new $(S^{(1)}, F^{(1)})$-monotone that is bounded by the generalized robustness $D_{\rm max} (\rho \| S^{(1)})$.
By following an analogous argument, we can show that if we choose relative entropy $R_2'(\rho) := D(\rho \| S^{(2)})$ as the $(S^{(2)}, F^{(2)})$-monotone, we obtain a new monotone that is upper bounded by the relative entropy $R_1'(\rho) \leq D(\rho \| S^{(1)})$.

Some properties of the original monotone is automatically inherited by the induced monotone, e.g. if $R_2$ is continuous, then so is $R_1$.
On the contrary, some properties require stronger assumptions: if $R_2$ is additive on tensor product states, then $R_1$ is only superadditive on tensor product states.
The following proposition shows that faithfulness is inherited under some mild additional constraints on $S^{(2)}$.

\begin{lemma}\label{proposition:monotone-faithful}
	Let $S^{(1)}, S^{(2)}$ be convex and compact sets.
	Furthermore, assume that $S^{(2)}$ is full-dimensional and $S^{(2)} \neq \mathcal{S}(\mathcal{H}^{(2)})$.
	Then, if $R_2$ is faithful, so is $R_1(\rho) = \sup R_2(\Tr_1 \bqty{\Lambda(\rho^{(1)} \otimes \mu^{(2)})})$, where the supremum is taken over all composite theories $(S, F)$, $\Lambda \in F$, and $\mu^{(2)} \in S^{(2)}$.
\end{lemma}
\begin{proof}
	Since we assume that $R_2$ is faithful, to prove the claim it is enough to show that for any $\rho \not\in S^{(1)}$, there is a CPTP map $\Lambda: \mathcal{S}(\mathcal{H}^{(1)}) \to \mathcal{S}(\mathcal{H}^{(2)})$ such that $\Lambda(S^{(1)}) \subseteq S^{(2)}$ and $\Lambda(\rho) \not\in S^{(2)}$.

	To this end, let us fix an arbitrary $\rho \not\in S^{(1)}$.
	Since $S^{(1)}$ is convex and compact, by the Hahn-Banach theorem there exists a Hermitian witness $W$ such that $\Tr W \rho < \inf_{\mu \in S^{(1)}} \Tr W \mu$.
     	Without loss of generality, we can assume $0 \leq W \leq \mathbf{1}$, since otherwise we can take
        \begin{align}
            W' = \frac{W + \abs{s} \mathbf{1}}{\norm{W}_{\infty} + \abs{s}},
        \end{align}
        where $s$ is the smallest eigenvalue of $W$.
	Similarly, without loss of generality we can assume that $\Tr W \rho < 1/2$, and $\inf_{\mu \in S^{(1)}} \Tr W \mu \geq 1/2$.
	Otherwise, let us denote $q^{*} = \inf_{\mu \in S^{(1)}} \Tr W \mu$, and note that $\Tr W \rho < q^{*}$.
	Let $W' = \frac{1}{2} \mathbf{1} + \epsilon \pqty{W - q^{*} \mathbf{1}}$, for some sufficiently small $\epsilon > 0$ such that $0 \leq W' \leq \mathbf{1}$.
	Then, we have $\Tr W' \rho = 1/2 + \epsilon \pqty{\Tr W\rho - q^{*}} < 1/2$ and $\inf_{\mu \in S^{(1)}} \Tr W' \mu = 1/2$.
	Now, since $S^{(2)}$ is full-dimensional, we can take a state $\tau$ in the interior of $S^{(2)}$.
	Take another state $\sigma \not\in S^{(2)}$, and let us define $\sigma_p = (1-p) \sigma + p \tau$, where $0 \leq p \leq 1$.
	Since $S^{(2)}$ is convex and closed, there exists $p^{*}$ such that $\sigma_p \in S^{(2)}$ if and only if $p \geq p^*$.
	Furthermore, since $\tau$ is in the interior of $S^{(2)}$, we have $p^* < 1$.
	Now, without any loss of generality, we can assume that $p^* = 1/2$, since otherwise we can take $\sigma' = (1-p^* + \delta) \sigma + (p^* - \delta) \tau, \tau' = (1-p^* - \delta) \sigma + (p^* + \delta) \tau$, and $\sigma'_p = (1-p) \sigma' + p \tau'$ instead, for some sufficiently small $\delta > 0$.

	Now, let $\Lambda (X) = \pqty{\Tr\bqty{(\mathbf{1} - W)X}} \sigma + \pqty{\Tr\bqty{WX}} \tau$.
	Since $0 \leq W \leq \mathbf{1}$ and $\sigma, \tau$ are valid states, $\Lambda$ is a CPTP map.
	Furthermore, for any $\mu \in S^{(1)}$ we have $\Lambda(\mu) = \Tr\bqty{(\mathbf{1} - W)\mu} \sigma + \Tr\bqty{W\mu} \tau = \sigma_{p = \Tr W\mu} \in S^{(2)}$, since $\Tr W\mu \geq 1/2$, showing $\Lambda(S^{(1)}) \subseteq S^{(2)}$.
	Finally, we have $\Lambda(\rho) \not\in S^{(2)}$, since $\Tr W\rho < 1/2$, showing the claim.
\end{proof}

Note that if we take the resource theory at party $2$ to be entanglement theory, the assumptions in \cref{proposition:monotone-faithful} are fulfilled.
This significantly strengthens the result in Ref.~\cite{Streltsov_2015};
not only can we measure coherence with entanglement, in fact we can measure \emph{any} convex resource with entanglement.
Furthermore, the resource conversion task is always achievable in \emph{any} RNG composite theory;
we can show that for any choice of composite free states $S$ and any $\rho \not\in S^{(1)}$, there is a CPTP map $\Lambda$ that preserves $S$ and $\Tr_1 \bqty{\Lambda\pqty{\rho^{(1)} \otimes \mu^{(2)}}} \not\in S^{(2)}$.

Note that in \cref{proposition:monotone}, it is crucial that we are taking supremum over all composite theories; otherwise, the monotone that we obtain might be trivial.
For example, if we only consider $\Fmin$ or LFOCC as free operations, then $R_1(\rho^{(1)})=0$ for any $\rho^{(1)}$, as these operations cannot convert one resource into another.
\begin{lemma}\label{proposition:lfocc}
  Let $S^{(2)}$ be a convex set and $\mu^{(2)} \in S^{(2)}$ be a free state.
  Then for any $\rho^{(1)}$ and any LFOCC protocol $\Lambda$, we have $\Tr_1 \Lambda(\rho^{(1)} \otimes \mu^{(2)}) \in S^{(2)}$.
\end{lemma}
\begin{proof}
  Intuitively, we rely on the fact that if $\Lambda$ is an LFOCC protocol, the local operation that is applied at party $2$ must be a free operation.
  We thus obtain the claim since all free operations must preserve the set of free states.
  To be precise, let $\Lambda(\rho^{(12)}) = \sum_{ij} \pqty{A_i \otimes B_{j|i}} \rho^{(12)} \pqty{A_i \otimes B_{j|i}}^\dagger$ be a $1$-way LFOCC from party 1 to party 2.
  We have
  \begin{align}
    \Tr_1 \Lambda(\rho^{(1)} \otimes \mu^{(2)})
    & =
      \Tr_1
      \sum_{ij}
      \pqty{A_i \otimes B_{j|i}}
      \rho^{(1)} \otimes \mu^{(2)}
      \pqty{A_i \otimes B_{j|i}}^\dagger
    \\
    & =
      \sum_{ij}
      \pqty{\Tr \bqty{A_i \rho^{(1)} A_i^\dagger}}
      B_{j|i}
      \mu^{(2)}
      B_{j|i}^\dagger.
  \end{align}
  Using the fact that $\mu'_i{}^{(2)} = \sum_j B_{j|i} \mu^{(2)} B_{j|i}^\dagger$ must be a free state and $S^{(2)}$ is convex, it follows that $\Tr_1 \Lambda(\rho^{(1)} \otimes \mu^{(2)})$ is in $S^{(2)}$.

  On the other hand, let $\Lambda'$ be a $1$-way LFOCC from party 2 to party 1.
  Then it naturally follows that $\Tr_1 \circ \Lambda'=\Tr_1 \otimes \tilde{\Lambda}^{(2)}$ with some free operation $\tilde{\Lambda}$ on party 2.
  Hence $\Tr_1 \Lambda(\rho^{(1)} \otimes \mu^{(2)})=\Phi(\mu)\in S^{(2)}$.
  Since every LFOCC protocol is composed of $1$-way LFOCC in both directions, the claim follows.
\end{proof}

\section{Further examples for remote certification}\label{supp:remote-certification}

In what follows, we will characterize the performance of LFOCC-based preprocessing and RNG protocols.

\subsubsection{LFOCC}

Suppose that LFOCC is the composite free operations.
First, note that we can always perform a locally free measurement at $A$, then send the outcome to $B$.
$B$ then simply guess according to $A$'s outcome.
This strategy shows that
\begin{align}
  \dist_{LFOCC}^{\epsilon} (\rho^{(A)} \| S^{(A)})
  &\geq
    D_{H, M_A}^\epsilon (\rho^{(A)} \| S^{(A)}),
\end{align}
where $M_A$ is the set of free measurements in $A$, and $D_{H, M}^\epsilon$ is the hypothesis testing relative entropy under measurements restricted to $M$.
We show that this is in fact the optimal LFOCC strategy when $A$ is governed by SIO coherence theory, even though $B$'s measurements are unrestricted.

\begin{lemma}\label{proposition:lfocc-certification-bound}
  Suppose $A$ is governed by the resource theory of SIO coherence.
  We have
  \begin{align}
    \dist_{LFOCC}^{\epsilon} (\rho^{(A)} \| S^{(A)})
    &\leq
      D_{H, M_A}^\epsilon (\rho^{(A)} \| S^{(A)}),
  \end{align}
\end{lemma}
\begin{proof}
  The idea is to show that for any measurement $\{P^{(B)}, \mathbf{1}-P^{(B)}\}$ on $B$, the effective measurement $\{\Lambda^\dagger (P^{(B)}), \mathbf{1}-\Lambda^\dagger (P^{(B)})\}$ that is performed on $A$ has to be free.
  First, let us analyze the case when $\Lambda$ is a $1$-round LFOCC.
  When the communication is going from Bob to Alice, the result holds trivially, so let us focus on the case where the communication goes from Alice to Bob.
  By definition, there has to be a set of Kraus operators $K_{ij} = A_i \otimes B_{j|i}$ such that $\Bqty{A_i}$ is an SIO and $\Bqty{B_{j|i}}$ is a real operation for each $i$.
  If we do a measurement with a POVM element $P^{(B)}$ on $B$, the effective measurement on $A$ is given by
  \begin{align}
    \Tr_B \Lambda^\dagger (\mathbf{1}_A \otimes P^{(B)})
     &=
      \Tr_B
      \sum_{ij}
      \pqty{A_i \otimes B_{j|i}}^\dagger
      \pqty{\mathbf{1}_A \otimes P^{(B)}}
      \pqty{A_i \otimes B_{j|i}}
    =
      \sum_i
      A_i^\dagger A_i
      \pqty{\sum_j \Tr B_{j|i}^\dagger P B_{j|i} }.
  \end{align}
  Since $\Bqty{A_i}$ is an SIO and in particular $A_i^\dagger A_i$ must be a diagonal matrix, the effective POVM element must also be diagonal.
  This implies that the effective measurement on $A$ is an SIO measurement.

  To see that allowing additional rounds will not change the statement, suppose that $\Lambda$ is a $3$-round LFOCC.
  By definition, there exists a set of Kraus operators $K_{ijkl} = A_{k|ij} A_i \otimes B_{l|ijk} B_{j|i}$, where $\Bqty{A_i}$ and $\Bqty{A_{k|ij}}$ are SIO's and $\Bqty{B_{j|i}}$ and $\Bqty{B_{l|ijk}}$ are real operations.
  Then,
  \begin{align}
    \Tr_B \Lambda^\dagger (\mathbf{1}_A \otimes P^{(B)})
    \nonumber
    & =
      \Tr_B
      \sum_{ijkl}
      \pqty{A_{k|ij} A_i \otimes B_{l|ijk} B_{j|i}}^\dagger
      \pqty{\mathbf{1}_A \otimes P^{(B)}}
      \pqty{A_{k|ij} A_i \otimes B_{l|ijk} B_{j|i}}
    \\
    & =
      \sum_{ijk}
      A_{k|ij}^\dagger A_i^\dagger A_i A_{k|ij}
      \pqty{
      \sum_l
      \Tr B_{j|i}^\dagger B_{l|ijk}^\dagger P^{(B)} B_{l|ijk} B_{j|i}
      }.
  \end{align}
  Again, since $A_i^\dagger A_i$ are diagonal and $\Bqty{A_{k|ij}}$ is an SIO, $A_{k|ij}^\dagger A_i^\dagger A_i A_{k|ij}$ must be diagonal and therefore the effective POVM element must also be diagonal.
  The argument for a general $n$-round protocol proceeds identically.
\end{proof}

The essential ingredient of the proof is that incoherent measurements are compatible with strictly incoherent operations;
if the POVM $\Bqty{P_i}_i$ is an incoherent measurement and $\Bqty{K_j}_j$ are Kraus operators for an SIO, then $\Bqty{K_j P_i K_j^\dagger}_{ij}$ is a POVM of some incoherent measurement.
This does not hold e.g. if $A$ can perform incoherent operations.
Nevertheless, this suggests that we can generalize \cref{proposition:lfocc-certification-bound} by substituting any theory that satisfies this compatibility condition.

Note that since SIO only includes incoherent measurements, we cannot distinguish any coherent state $\rho^{(A)}$ from incoherent states with better than trivial performance.
Combined with \cref{proposition:lfocc-certification-bound}, this implies that LFOCC protocols provides no improvement over the trivial strategy.

\subsubsection{Optimal performance beyond LFOCC}

Now, let us show that the bound in \cref{proposition:remote-certification-optimal} is achievable under an additional assumption on the local resource theories.
For example, suppose $S^{(A)} \subseteq S^{(B)}$.
Then for any choice of composite free state $S$, the composite RNG theory $(S, \RNG(S))$ achieves the optimal performance.
\begin{lemma}~\label{proposition:resource-certification-optimal}
  Suppose that $S^{(A)} \subseteq S^{(B)}$.
  Then, for any RNG composite theory $(S, \RNG(S))$, we have
  \begin{align}
\dist_{RNG(S)}^\epsilon (\rho^{(A)} \| S^{(A)})
    &=
      D_{H}^\epsilon (\rho^{(A)} \| S^{(A)}),
  \end{align}
\end{lemma}
\begin{proof}
  Showing that $D_{H}^\epsilon (\rho^{(A)} \| S^{(A)})$ is achievable is sufficient for the proof.
  To this end, observe that for any fixed $\mu^{(A)} \in S^{(A)}$, the map $\Lambda(X^{(A)} \otimes Y^{(B)}) = \mu^{(A)} \otimes X^{(B)}$ is in $\RNG(S)$ since $S^{(A)} \subseteq S^{(B)}$.
  Also, note that $\Lambda^\dagger(\mathbf{1}_A \otimes P^{(B)}) = P^{(A)} \otimes \mathbf{1}_B$.
  Combining these two observations, we obtain the claim.
\end{proof}

The conditions for \cref{proposition:resource-certification-optimal} are satisfied when $A$ is the resource theory of SIO coherence and $B$ is the resource theory of imaginarity.
This shows that there is a maximal separation in the performance between LFOCC and the optimal resource theoretic protocols.

\end{document}